\documentclass[11pt]{article} 
\usepackage[margin=1in]{geometry} 
\pdfoutput=1

\usepackage[usenames,dvipsnames,svgnames,table]{xcolor} 

\usepackage{amsmath,amsfonts,amssymb,amsthm}
\usepackage{mathtools}
\usepackage{braket}
\usepackage{diagbox}
\usepackage{thmtools,thm-restate} 

\usepackage{tikz,tikz-qtree} 
\usepackage{wrapfig} 

\usepackage{mmap}
\usepackage[T1]{fontenc}

\usepackage[utf8]{inputenc}
\usepackage{csquotes}
\usepackage[english]{babel}

 \usepackage[
     activate={true,nocompatibility},
     final, 
     tracking=true,
     kerning=true,
     spacing=true,
     factor=1100,
     stretch=10,
     shrink=10
     ]{microtype}
 \microtypecontext{spacing=nonfrench}

\usepackage[backend=biber,
              style=alphabetic,
        maxbibnames=99,
      minalphanames=3,
      maxalphanames=4,
            backref=true]{biblatex}

\usepackage{hyperref}
\hypersetup{
    pdftitle={Adversary Lifting}, 
    pdfauthor={Anurag Anshu, Shalev Ben{-}David, Srijita Kundu}, 
    colorlinks=true, 
    linkcolor=blue, 
    citecolor=blue, 
    urlcolor=green 
}

\allowdisplaybreaks

\AtEveryBibitem{
 \clearlist{address}
 \clearfield{date}
 \clearlist{location}
 \clearfield{month}
 \clearfield{series}
 \clearfield{pages}
 \clearlist{organization}
 \clearfield{number}
 \clearlist{language}

 \ifentrytype{book}{}{
  \clearlist{publisher}
  \clearname{editor}
  \clearfield{issn}
  \clearfield{isbn}
  \clearfield{volume}
 }
}

\DeclareFieldFormat{eprint:eccc}{
\mkbibacro{ECCC}\addcolon\space\ifhyperref
    {\href{http://eccc.hpi-web.de/report/#1/}{\nolinkurl{#1}}}
    {\nolinkurl{#1}}}
\DeclareFieldAlias{eprint:ECCC}{eprint:eccc}
\DeclareFieldFormat{eprint:iacr}{Cryptology ePrint\addcolon\space\ifhyperref
    {\href{https://eprint.iacr.org/#1}{\nolinkurl{#1}}}
    {\nolinkurl{#1}}}
\DeclareFieldAlias{eprint:IACR}{eprint:iacr}
\DeclareFieldFormat{eprint:arxiv}{arXiv\addcolon\space\ifhyperref
    {\href{https://arxiv.org/abs/#1}{\nolinkurl{#1}}}
    {\nolinkurl{#1}}}

\DeclareFieldFormat[article]{title}{#1}
\DeclareFieldFormat[inproceedings]{title}{#1}
\DeclareFieldFormat[online]{title}{#1}

\renewbibmacro{in:}{}

\newcommand{\lName}{1}

\newcommand{\donothing}[1]{#1}

\newcommand{\JACM}{\if\lName1\donothing{Journal of the {ACM}}\else{JACM}\fi}
\newcommand{\SICOMP}{\if\lName1\donothing{{SIAM} Journal on Computing}\else{SICOMP}\fi}
\newcommand{\ToC}{\if\lName1\donothing{Theory of Computing}\else{ToC}\fi}
\newcommand{\ToCGS}{\if\lName1\donothing{Theory of Computing Graduate Surveys}\else{ToC}\fi}
\newcommand{\TOCT}{\if\lName1\donothing{{ACM} Transactions on Computation Theory}\else{TOCT}\fi}
\newcommand{\CCjournal}{\if\lName1\donothing{Computational Complexity}\else{CC}\fi}
\newcommand{\CJTCS}{\if\lName1\donothing{Chicago Journal of Theoretical Computer Science}\else{CJTCS}\fi}

\newcommand{\TCS}{\if\lName1\donothing{Theoretical Computer Science}\else{TCS}\fi}
\newcommand{\IPL}{\if\lName1\donothing{Information Processing Letters}\else{IPL}\fi}
\newcommand{\JCSS}{\if\lName1\donothing{Journal of Computer and System Sciences}\else{JCSS}\fi}

\newcommand{\RSA}{\if\lName1\donothing{Random Structures and Algorithms}\else{RSA}\fi}
\newcommand{\JCTA}{\if\lName1\donothing{Journal of Combinatorial Theory, Series A}\else{JCTA}\fi}
\newcommand{\JCTB}{\if\lName1\donothing{Journal of Combinatorial Theory, Series B}\else{JCTB}\fi}
\newcommand{\PJM}{\if\lName1\donothing{Pacific Journal of Mathematics}\else{PJM}\fi}
\newcommand{\QICjournal}{\if\lName1\donothing{Quantum Information and Computation}\else{QIC}\fi}
\newcommand{\IJQI}{\if\lName1\donothing{International Journal of Quantum Information}\else{IJQI}\fi}
\newcommand{\PRA}{\if\lName1\donothing{Physical Review A}\else{PRA}\fi}
\newcommand{\PRL}{\if\lName1\donothing{Physical Review Letters}\else{PRL}\fi}

\DefineBibliographyStrings{english}{%
  backrefpage = {p{.}},
  backrefpages = {pp{.}},
}

\newtheorem{theorem}{Theorem}

\newtheorem{lemma}[theorem]{Lemma}

\newtheorem{corollary}[theorem]{Corollary}
\newtheorem{definition}[theorem]{Definition}

\newtheorem{fact}[theorem]{Fact}

\theoremstyle{definition}

\newcommand{\eq}[1]{\hyperref[eq:#1]{(\ref*{eq:#1})}}
\renewcommand{\sec}[1]{\hyperref[sec:#1]{Section~\ref*{sec:#1}}}
\newcommand{\thm}[1]{\hyperref[thm:#1]{Theorem~\ref*{thm:#1}}}
\newcommand{\lem}[1]{\hyperref[lem:#1]{Lemma~\ref*{lem:#1}}}
\newcommand{\defn}[1]{\hyperref[def:#1]{Definition~\ref*{def:#1}}}
\newcommand{\prop}[1]{\hyperref[prop:#1]{Proposition~\ref*{prop:#1}}}
\newcommand{\cor}[1]{\hyperref[cor:#1]{Corollary~\ref*{cor:#1}}}
\newcommand{\fig}[1]{\hyperref[fig:#1]{Figure~\ref*{fig:#1}}}
\newcommand{\tab}[1]{\hyperref[tab:#1]{Table~\ref*{tab:#1}}}
\newcommand{\alg}[1]{\hyperref[alg:#1]{Algorithm~\ref*{alg:#1}}}
\newcommand{\app}[1]{\hyperref[app:#1]{Appendix~\ref*{app:#1}}}
\newcommand{\conj}[1]{\hyperref[conj:#1]{Conjecture~\ref*{conj:#1}}}
\newcommand{\chap}[1]{\hyperref[chap:#1]{Chapter~\ref*{chap:#1}}}
\newcommand{\fct}[1]{\hyperref[fct:#1]{Fact~\ref*{fct:#1}}}

\DeclareMathOperator{\poly}{poly}
\DeclareMathOperator{\polylog}{polylog}
\newcommand{\tO}{\tilde{O}}
\newcommand{\tOmega}{\tilde{\Omega}}

\DeclareMathOperator{\Dom}{Dom}

\newcommand{\B}{\{0,1\}}

\newcommand{\X}{\mathcal{X}}
\newcommand{\Y}{\mathcal{Y}}
\newcommand{\Z}{\mathcal{Z}}
\newcommand{\CC}{\textsc{CC}}

\DeclareMathAlphabet{\mathbbold}{U}{bbold}{m}{n}

\DeclareMathOperator{\Tr}{Tr}

\DeclareMathOperator{\bR}{\mathbb{R}}
\DeclareMathOperator{\bN}{\mathbb{N}}

\newcommand{\cB}{\mathcal{B}}

\newcommand{\tX}{\widetilde{X}}
\newcommand{\tY}{\widetilde{Y}}

\newcommand{\norm}[1]{{\left\lVert#1\right\rVert}}

\newcommand{\OR}{\textsc{OR}}
\newcommand{\PrOR}{\textsc{PrOR}}
\newcommand{\AND}{\textsc{AND}}

\newcommand{\PARITY}{\textsc{Parity}}

\newcommand{\Disj}{\textsc{Disj}}
\newcommand{\IP}{\textsc{IP}}

\DeclareMathOperator{\D}{D}
\DeclareMathOperator{\R}{R}
\DeclareMathOperator{\Q}{Q}

\DeclareMathOperator{\bs}{bs}
\DeclareMathOperator{\fbs}{fbs}
\DeclareMathOperator{\cbs}{cbs}
\DeclareMathOperator{\cfbs}{cfbs}
\DeclareMathOperator{\Adv}{Adv}

\DeclareMathOperator{\M}{M} 
\DeclareMathOperator{\CAdv}{CAdv}
\DeclareMathOperator{\RCC}{RCC}
\DeclareMathOperator{\QCC}{QCC}

\DeclareMathOperator{\QICZ}{QICZ}
\DeclareMathOperator{\IC}{IC}
\DeclareMathOperator{\QIC}{QIC}

\DeclareMathOperator{\HQIC}{HQIC}
\DeclareMathOperator{\crit}{crit}

\DeclareMathOperator{\adeg}{\widetilde{\deg}}

\begin{document}

\title{On Query-to-Communication Lifting
for Adversary Bounds}

\author{
Anurag Anshu\\
\small University of California, Berkeley\\
\small \texttt{anuraganshu@berkeley.edu}
\and
Shalev Ben{-}David\\
\small University of Waterloo\\
\small \texttt{shalev.b@uwaterloo.ca}
\and
Srijita Kundu\\
\small National University of Singapore \\
\small \texttt{srijita.kundu@u.nus.edu}
}

\date{}
\maketitle

\begin{abstract}
We investigate query-to-communication lifting theorems
for models related to the quantum adversary bounds.
Our results are as follows:
\begin{enumerate}
\item We show that the classical adversary bound
lifts to a lower bound on randomized communication complexity
with a constant-sized gadget. We also
show that the classical adversary bound
is a strictly stronger lower bound technique
than the previously-lifted measure
known as critical block sensitivity,
making our lifting theorem one of the strongest lifting
theorems for randomized communication complexity
using a constant-sized gadget.
\item Turning to quantum models, we show a connection
between lifting theorems for quantum adversary bounds
and secure 2-party quantum computation in a certain
``honest-but-curious'' model. Under the assumption
that such secure 2-party computation is impossible,
we show that a simplified version of the positive-weight
adversary bound lifts to a quantum communication lower bound
using a constant-sized gadget.
We also give an unconditional lifting theorem
which lower bounds bounded-round quantum communication protocols.
\item Finally, we give some new results in query complexity.
We show that the classical adversary and the positive-weight
quantum adversary are quadratically related. We also show
that the positive-weight quantum adversary is never larger
than the square of the approximate degree. Both relations
hold even for partial functions.
\end{enumerate}
\end{abstract}

{\small\tableofcontents}

\section{Introduction}

Communication complexity is an important model of computation
with deep connections to many parts of theoretical computer science
\cite{KN96}. In communication complexity, two parties,
called Alice and Bob, receive inputs $x$ and $y$ from
sets $\X$ and $\Y$ respectively, and wish to compute
some joint function $F\colon\X\times\Y\to\B$ on their inputs.
Alice and Bob cooperate together, and
their goal is to minimize the number of bits they must exchange
before determining $F(x,y)$.

Recently, a lot of attention has been devoted to
connections between communication complexity and query complexity.
In particular, query-to-communication ``lifting'' theorems are powerful tools which convert lower bounds in query complexity into lower bounds in communication complexity in a black-box manner.
Since query lower bounds are typically much easier to prove
than communication lower bounds, these tools are highly useful
for the study of communication complexity, and often
come together with new communication complexity results
(such as separations between different communication
complexity models). For example, see
\cite{Goo15,GLM+16,GPW18,GPW20,CFK+19}.

Lifting theorems are known for many models of computation,
including deterministic \cite{GPW18} and randomized
\cite{GPW20} algorithms. Notably, however, a lifting theorem
for quantum query complexity is not known; the closest
thing available is a lifting theorem for approximate degree
(also known as the polynomial method), which lifts to approximate
logrank \cite{She11}. This allows quantum query lower bounds
proved via the polynomial method to be turned into quantum
communication lower bounds, but a similar statement is not
known even for the positive-weight quantum adversary method
\cite{Amb02,SS06}.

In this work, we investigate lifting theorems for the adversary
method and related models. We prove a lifting theorem
for a measure called the classical adversary bound.
For the quantum adversary method, we show that there is a
surprising connection with the cryptographic notion
of secure 2-party computation. Specifically,
we show that a lifting theorem for a simplified
version of the positive-weight adversary method follows
from a plausible conjecture regarding the impossibility
of secure 2-party computation in a certain ``honest but curious''
quantum model. We also prove an unconditional
lifting theorem which lower bounds bounded-round
quantum algorithms.

Finally, we prove some query complexity results that may be
of independent interest: first, a quadratic relationship
between the positive-weight adversary bound and the
classical adversary bound; and second, we show that
the positive-weight adversary bound can never be larger than
the square of the approximate degree. This means that
the (positive) adversary method can never beat the polynomial
method by more than a quadratic factor. These results
hold even for partial functions.

\subsection{Lifting theorems}

The statement of a lifting theorem typically has the following form:
\[\M^{cc}(f\circ G)=\Omega(\M(f)).\]
Here $G\colon\X\times\Y\to\B$
is a (fixed) communication complexity function, called
a ``gadget'', which typically has low communication cost;
$f\colon\B^n\to\B$ is an arbitrary Boolean function; $\M(\cdot)$ is a
measure in query complexity, representing the cost of computing
the function $f$ in query complexity; and $\M^{cc}(\cdot)$ is a
measure in communication complexity. The notation $f\circ G$
denotes the block-composition of $f$ with $G$.
This is a communication complexity function defined
as follows: Alice gets input $(x_1,x_2,\dots, x_n)\in\X^n$,
Bob gets input $(y_1,y_2,\dots,y_n)\in\Y^n$, and they
must output $f(G(x_1,y_1),G(x_2,y_2),\dots,G(x_n,y_n))$.
Hence $f\circ G$ is a function with signature $\X^n\times\Y^n\to\B$.

There are two primary types of lifting theorems:
those that work with a constant-sized gadget $G$
(independent of $f$), and those that
work with a gadget $G$ whose size logarithmic in the
input size $n$ of $f$.\footnote{Sometimes, lifting theorems
use a gadget $G$ which is large -- polynomial in $n$ --
but which can still be computed using $O(\log n)$ communication.}
The latter type tend to be much more prevalent;
recent lifting theorems for deterministic and randomized
communication complexities all use log-sized or larger gadgets 
\cite{GPW18,WYY17,CKLM19,GPW20,CFK+19}. We remark, however,
that even with log-sized or larger gadgets, lifting theorems
are highly nontrivial to prove: a lifting theorem for
$\mathsf{BPP}$, which lifts randomized query lower bounds
to randomized communication lower bounds, was only established
in the last few years, while an analogous result
for $\mathsf{BQP}$ remains an open problem.

Lifting theorems which work with a constant-sized gadget
are even harder to prove,
but often turn out to be much more useful.
The reason is that common function families, like
disjointness (which we denote $\Disj_n$)\footnote{In the disjointness
function, Alice and Bob receive $n$-bit strings $x$ and $y$ and must
output $1$ if and only if there exists an index $i\in[n]$
such that $x_i=y_i=1$.}
or inner product (which we denote $\IP_n$)\footnote{In the
inner product function, Alice and Bob receive $n$-bit
strings $x$ and $y$, and must compute inner product
of those strings over $\mathbb{F}_2$.}, are
\emph{universal}. This means for every communication function
$G\colon\X\times\Y\to\B$, its communication matrix (that is,
its truth table) is a submatrix of the communication matrix
of (a sufficiently large instance of) the $\Disj$
function. In other words, every communication function
is a sub-function of $\Disj_k$ and $\IP_k$
for sufficiently large $k$.
If the size of $G$ is constant, then it is necessarily
contained in a $\Disj$ function of \emph{constant}
size (and similarly for $\IP$). Hence
lifting with any constant-sized gadget $G$
is enough to guarantee a lifting theorem with
a constant-sized disjointness gadget and constant-sized
inner product gadget (and similarly for every
other universal function family).
In short, lifting with a constant-sized gadget
implies lifting with (almost) any gadget of your choice.

In particular, a lifting theorem with a constant-sized gadget
immediately implies a lower bound for $\Disj_n$
and $\IP_n$ themselves. To see this,
suppose we had a lifting theorem 
\[\M^{cc}(f\circ G)=\Omega(\M(f))\]
for all Boolean functions $f$ and a fixed (constant-sized)
communication gadget $G$. Then $G$ is a sub-function
of $\Disj_k$ and of $\IP_k$ for some constant $k$.
Note that $\Disj_n=\OR_{n/k}\circ \Disj_k$
and that $\IP_n=\PARITY_{n/k}\circ \IP_k$.
Hence we get $\M^{cc}(\Disj_n)=\Omega(\M(\OR_{n/k}))$
and $\M^{cc}(\IP_n)=\Omega(\M(\PARITY_{n/k}))$.
Since $k$ is constant, this can potentially give
lower bounds on $\M^{cc}(\Disj_n)$ and on $\M^{cc}(\IP_n)$
that are tight up to constant factors, depending
on the measures $\M^{cc}(\cdot)$ and $\M(\cdot)$.

There have only been a handful of lifting theorems
which work with constant-sized gadgets.
One such result follows from Sherstov's work
for approximate degree and related measures \cite{She11}.
The part of that work which is most relevant to us
is the lifting of approximate degree to lower bounds
on approximate logrank, and hence on the quantum
communication complexity of the lifted function.
Sherstov's work means that if one can prove
a quantum lower bound for a query function $f$
using the \emph{polynomial method} \cite{BBC+01}, then this lower
bound will also apply to the quantum
communication complexity of $f\circ G$, where $G$
is a constant sized gadget. Such a lifting theorem
is not known to hold for the adversary methods
\cite{Amb02,SS06}, however (not even with a log-sized gadget).

Another lifting theorem with a constant-sized gadget
appears in \cite{HN12,GP18}. There, a query
measure called \emph{critical block sensitivity} \cite{HN12}
is lifted to a lower bound on randomized communication complexity.

\subsection{Adversary methods}

The quantum adversary bounds are extremely useful methods
for lower bounding quantum query complexity. The original
adversary method was introduced by Ambainis \cite{Amb02}.
It was later generalized in several ways, which were shown
to all be equivalent \cite{SS06}, and are known as the
positive-weight adversary bound, denoted $\Adv(f)$.
This bound has many convenient properties: it has
many equivalent formulations (among them a semidefinite
program), it is reasonably easy to use in practice,
and it behaves
nicely under many operations, such as composition.
The positive-weight adversary bound is one of the most
commonly used techniques for lower bounding quantum
query complexity.

A related measure is called the negative-weight
adversary bound, introduced in \cite{HLS07},
which we denote by $\Adv^{\pm}(f)$. This is a
strengthening of the positive adversary bound, and satisfies
$\Adv^{\pm}(f)\ge\Adv(f)$ for all (possibly partial)
Boolean functions $f$. Surprisingly, in \cite{Rei11,LMR+11},
it was shown that the negative-weight adversary is actually
\emph{equal} to quantum query complexity up to constant factors.

The quantum adversary methods have no known communication complexity
analogues. However, that by itself does not rule out a lifting
theorem: one might still hope to lift $\Adv(f)$ or
$\Adv^{\pm}(f)$ to lower bounds on quantum communication
complexity, similar to how critical block sensitivity
$\cbs(f)$ was lifted \cite{HN12,GP18}. Unfortunately,
no such lifting theorems are currently known, not even for
the positive-weight adversary method, and not even with
a large gadget size.

Interestingly, it is possible to define
a lower bound technique for \emph{randomized}
algorithms which is motivated by the (positive)
quantum adversary method. This measure was first introduced
in \cite{Aar06,LM08}, and different variants of it have been
subsequently studied \cite{AKPV18}. Here, we use the largest
of these variants, which we denote by $\CAdv(f)$
(in \cite{AKPV18}, it was denoted by $\operatorname{CMM}(f)$).
In \cite{AKPV18}, it was shown that
for total functions $f$, $\CAdv(f)$ is (up to constant
factors) equal to a measure called fractional block sensitivity,
which we denote $\fbs(f)$. However, for partial functions,
there can be a large separation between the two measures.
For more on fractional block sensitivity,
see \cite{Aar08,KT16}.

\subsection{Our contributions}

\subsubsection*{Lifting the classical adversary}

Our first contribution is a lifting theorem for the classical
adversary bound $\CAdv(f)$. We lift it to a lower bound on randomized
communication complexity using a constant-sized gadget.

\begin{theorem}\label{thm:CAdv_lift}
There is an explicitly given function
$G\colon \X\times\Y\to\B$ such that
for any (possibly partial) Boolean function $f$,
\[\RCC(f\circ G)=\Omega(\CAdv(f)).\]
\end{theorem}

Here $\RCC(f\circ G)$
denotes the randomized communication complexity of $f\circ G$
with shared randomness.
We note that \cite{HN12,GP18} provided a lifting
theorem that has a similar form, only with the measure $\cbs(f)$
in place of $\CAdv(f)$. To compare the two theorems,
we should compare the two query measures. We have the following
theorem.

\begin{lemma}\label{lem:CAdv_cbs}
For all (possibly partial) Boolean functions $f$,
$\CAdv(f)=\Omega(\cbs(f))$. Moreover, there is a family
of total functions $f$ for which $\CAdv(f)=\Omega(\cbs(f)^{3/2})$. 
\end{lemma}

\lem{CAdv_cbs} says that $\CAdv(f)$ is a strictly stronger
lower bound technique than $\cbs(f)$, and hence \thm{CAdv_lift}
is stronger than the lifting theorem of \cite{GP18}.
This makes \lem{CAdv_cbs} one of the strongest known
lifting theorems for randomized communication complexity
which works with a constant-sized gadget\footnote{Sherstov's
lifting theorem for approximate degree \cite{She11} also
works with a constant-sized gadget, and is incomparable
to our result as a lower bound technique for randomized
communication complexity.}.

We note that the lifting theorem of \cite{GP18} for
the measure $\cbs(f)$ also works when $f$ is a \emph{relation},
which is a more general setting than partial functions;
indeed, most of their applications for the lifting theorem
were for relations $f$ rather than functions.
We extend \thm{CAdv_lift} to relations as well, and also
show that $\CAdv(f)=\Omega(\cbs(f))$ for all relations.
In fact, it turns out that for partial functions, $\CAdv(f)$
is equal to a fractional version of $\cbs(f)$, which we denote
$\cfbs(f)$; however, for relations, $\CAdv(f)$ is a stronger
lower bound technique than $\cfbs(f)$ (which in turn
is stronger than $\cbs(f)$). We also note that our
techniques for lifting $\CAdv(f)$ are substantially different
from those of \cite{HN12,GP18}.

\subsubsection*{Lifting quantum measures}

Our first quantum result says that $\CAdv(f)$ lifts
to a lower bound on bounded-round quantum communication
protocols. This may seem surprising, as $\CAdv(f)$
does not lower bound quantum algorithms in query complexity;
however, one can show that $\CAdv(f)$ does lower bound
\emph{non-adaptive} quantum query complexity, or even
quantum query algorithms with limited adaptivity.
This motivates the following result.

\begin{theorem}
There is an explicitly given function $G\colon\X\times\Y\to\B$
such that for any (possibly partial) Boolean function $f$,
\[\QCC^r(f\circ G)=\Omega(\CAdv(f)/r^2).\]
Here $\QCC^r(\cdot)$ denotes the quantum communication
complexity for an $r$-round quantum protocol with
shared entanglement.
\end{theorem}

We note that since any $r$-round protocol has communication
cost at least $r$, we actually get a lower bound of
$\CAdv(f)/r^2+r$. Minimizing over $r$ yields
a lower bound of $\CAdv(f)^{1/3}$ even on unbounded-round
protocols. This may not seem very useful, since
$\CAdv(f)^{1/3}$ is smaller than $\adeg(f)$,
a measure we know how to lift \cite{She11}. However,
we can generalize this result to relations. For
relations, we do not know how to compare $\CAdv(f)^{1/3}$
to $\adeg(f)$, and therefore our lifting theorem gives
something new, even in the unbounded-round setting.

\begin{corollary}
There is an explicitly given function $G\colon\X\times\Y\to\B$
such that for any relation $f$,
\[\QCC(f\circ G)=\Omega(\CAdv(f)^{1/3}),\]
where $\QCC$ denotes the quantum communication complexity
with shared entanglement.
\end{corollary}

We next turn our attention to lower bounding unbounded-round
quantum communication protocols by lifting a quantum
adversary method. Instead of aiming for the positive-weight
adversary bound, we instead work with a simplified version,
studied in \cite{ABK+20}, which we denote
$\Adv_1(f)$. This measure is a restriction of $\Adv$ to a
pairs of inputs with a single bit of difference.

We have $\Adv_1(f)\le \Adv(f)$, and \cite{ABK+20}
showed that $\Adv_1(f)=O(\adeg(f))$.
However, their proof of the latter is tricky,
and we do not use it here; we give a direct lifting
of $\Adv_1(f)$ (under a certain assumption),
and we argue that the techniques we use are likely
to generalize to lifting $\Adv(f)$ in the future.

We prove the following theorem.

\begin{theorem}\label{thm:Adv_one_lift}
For any (possibly partial) Boolean function $f$
and any communication function $G$ which contains
both $\AND_2$ and $\OR_2$ as subfunctions, we have
\[\QCC(f\circ G)=\Omega(\Adv_1(f)\QICZ(G)).\]
This also holds for relations $f$.
\end{theorem}

At first glance, this theorem might look very strong: not
only does it lift the simplified adversary bound for a single
gadget $G$, it even does so for all $G$ with a dependence on $G$.
Unfortunately, there is a catch: the measure $\QICZ(G)$
may be $0$ for some communication
functions $G$. In fact, we cannot rule out the possibility
that $\QICZ(G)=0$ for \emph{all} communication functions $G$,
in which case \thm{Adv_one_lift} does not say anything. On
the other hand, note that if $\QICZ(G)>0$ for even a single
function $G$, then \thm{Adv_one_lift} gives a lifting theorem
for $\Adv_1(f)$ with a constant-sized gadget,
which works even for relations (since
that single $G$ can have no dependence on the input size of $f$).

We give an interpretation of the measure $\QICZ(G)$ in terms
of a cryptographic primitive called secure 2-party computation.
In such a primitive, Alice and Bob want to compute a function
$G$ on their inputs $x$ and $y$,
but they do not want to reveal their
inputs to the other party. Indeed, Alice wants to hide
everything about $x$ from Bob and Bob wants to hide
everything about $y$ from Alice, with the exception
of the final function value $G(x,y)$ (which they are
both expected to know at the end of the protocol).
We also seek information-theoretic security:
there are no limits on the computational power of Alice and Bob.
Since we are interested in a quantum version, we will
allow Alice and Bob to exchange quantum communication rather
than classical communication, potentially with shared entanglement.

Secure 2-party computation is known to be impossible
in general, even quantumly \cite{Lo97,Col07,BCS12,FKS+13,SSS14}.
However, in our case, we care about an ``honest but curious'' version
of the primitive, in which Alice and Bob trust each other to
execute the protocol faithfully, but they still do not
trust each other not to try to learn the others' input.
In the quantum setting, it is a bit difficult to define such
an honest-but-curious model: after all, if Alice and Bob
are honest, they might be forbidden by the protocol from
ever executing intermediate measurements, and the protocol
might even tell them to ``uncompute'' everything except
for the final answer, to ensure all other information
gets deleted. Hence it would seem that honest parties can trivially
do secure 2-party computation.

The way we will define quantum secure two-party computation
in the honest-but-curious setting will be analogous
to the information-based classical definition
(see, for example, \cite{BW15}). Classically,
the information leak that Alice and Bob must suffer
in an honest execution of the best possible
protocol is captured by $\IC(G)$, the information
cost of the function $G$. The measure $\IC(G)$
is the amount Alice learns about Bob's input
plus the amount Bob learns about Alice's input,
given the best possible protocol and the worst possible
distribution over the inputs; we note that this measure
includes the value of $G(x,y)$ as part of what Alice and Bob
learn about each others' inputs, whereas secure
two-party computation does not count learning $G(x,y)$
as part of the cost, but this is only a difference of at most
$2$ bits of information (one on Alice's side and one on Bob's side);
hence, up to an additive factor of $2$, $\IC(G)$
captures the information leak necessary in a two-party
protocol computing $G$.

For a quantum version of this, we will use $\QIC$,
a measure which is a quantum analogue
of $\IC$ and which was introduced in \cite{Tou15}.
However, we note that if Alice and Bob send the same bit
$G(x,y)$ back and forth $n$ times, this will add $\Theta(n)$
to the value of $\QIC$ for that protocol, due to subtleties
in the definition of $\QIC$ (this does not occur classically
with $\IC$). Hence, in the quantum setting, $\QIC$ does not
capture the two-party information leak as cleanly as $\IC$
did classically.

Instead, we modify the definition of $\QIC$ to a measure
we denote $\QICZ(G)$.
For this measure,
Alice and Bob want a protocol $\Pi$ such that for any
distribution $\mu$ that has support only on $0$-inputs
(or only on $1$-inputs), $\QIC(\Pi,\mu)$ is small.
This will ensure that Alice and Bob learn nothing
about each others' inputs when conditioned on the output
of the function. The two-party secure
computation question then becomes:
does such a secure protocol $\Pi$ exists for computing
any fixed communication function $G$?

Intuitively, we believe that the answer should be no,
at least for some communication functions $G$. This would
align with the known impossibility of various types of
secure 2-party quantum computation, though none of those
impossibility results seem to apply to our setting.
Interestingly, we have the following lemma, which
follows directly form the way we define $\QICZ(G)$.

\begin{lemma}
Suppose that our version of secure 2-party quantum computation
is impossible for a communication function $G$ which contains
both $\AND$ and $\OR$ as sub-functions. Then
$\QICZ(G)>0$, and hence $\Adv_1(\cdot)$ lifts to a
quantum communication lower bound with the gadget $G$.
\end{lemma}

\subsubsection*{New query relations}

Finally, our study of the classical adversary bound
led to some new relations in query complexity that are likely
to be of independent interest.

\begin{theorem}\label{thm:Adv_adeg}
For all (possibly partial) Boolean functions $f$,
\[\Adv(f)= O(\adeg(f)^2).\]
\end{theorem}

Here $\adeg(f)$ is the approximate degree of $f$
to bounded error.\footnote{This is the minimum
degree of an $n$-variate real polynomial $p$
such that $|p(x)|\in[0,1]$ for all $x\in\B^n$
and such that $|p(x)-f(x)|\le 1/3$ for all $x$
in the domain of $f$.}
This relationship is interesting, as it says that the
positive-weight adversary method can never beat the polynomial
method by more than a quadratic factor. Conceivably,
this can even be used as a lower bound technique
for the approximate degree of Boolean functions
(which is a measure that is often of interest even
apart from quantum lower bounds). In fact, we prove
a strengthening of \thm{Adv_adeg}.

\begin{theorem}
For all (possibly partial) Boolean functions $f$,
\[\adeg_\epsilon(f)\ge\frac{\sqrt{(1-2\epsilon)\CAdv(f)}}{\pi}.\]
\end{theorem}

This version of the theorem is stronger, since
$\Adv(f)\le\CAdv(f)$. Finally, we prove a quadratic
relationship between the classical and quantum (positive-weight)
adversary bounds.

\begin{theorem}
For all (possibly partial) Boolean functions $f$,
\[\Adv(f)\le\CAdv(f)\le 2\Adv(f)^2.\]
\end{theorem}

We note that all of these new relations hold even for partial
functions. This is unusual in query complexity, where most
relations hold only for total functions, and where
most pairs of measures can be exponentially separated
in the partial function setting.

\subsection{Our Techniques}

We introduce several new techniques that we believe will
be useful in future work on adversary methods in communication
complexity.

\subsubsection*{A lifting framework for adversary methods}

One clear insight we contribute in this work is that
lifting theorems for adversary method can be fruitfully
attacked in a ``primal'' way, and using information cost.
To clarify, our approach is to take a protocol $\Pi$
for the lifted function $f\circ G$, and to convert it
into a solution to the primal (i.e.\ minimization) program
for the target adversary bound of $f$.

The primal program for an adversary method
generally demands a non-negative weight $q(z,i)$
for each input string $z\in\B^n$ and each index $i\in[n]$,
such that a certain feasibility constraint is satisfied
for each pair $(z,w)$ with $f(z)\ne f(w)$,
and such that $\sum_{i\in[n]}q(z,i)$ is small for each
input $z$. Our approach is to use an information cost
measure to define $q(z,i)$, where the information
is measured against a distribution $\mu_z$ over
$n$-tuples of inputs to $G$ that evaluate to $z$,
and where we only measure the information
transmitted by the protocol about the $i$-th input to $G$,
conditioned on the previous bits.

We show that this way of getting a solution
to the (minimization version of) the adversary bound for $f$
using a communication protocol for $f\circ G$ suffices
for lifting $\CAdv$ to a randomized communication lower bound
(with a constant-sized gadget), and that it also suffices for getting
some quantum lifting theorems.

\subsubsection*{Product-to-sum reduction}

One of the main tools we use in the proof of
the lifting theorem for $\Adv_1$ is what we call a product-to-sum
reduction for quantum information cost. We show that
if there is a protocol $\Pi$ which computes
some communication function $F$ such that the product
$\QIC(\Pi,\mu_0)\cdot\QIC(\Pi,\mu_1)$ is small
(where $\mu_0$ and $\mu_1$ are distributions over
$0$- and $1$-inputs to $F$),
then there is also a protocol $\Pi'$ which also computes
$F$ and for which $\QIC(\Pi,\mu_0) + \QIC(\Pi,\mu_1)$
is small. In particular, a lower bound for the latter
measure implies a lower bound for the former.
This is useful because the sum (or maximum) of the
two quantum information costs is a natural operation
on quantum information measures (to which lower bound tools
may apply), while the product is not; yet
the product of these information measures arises
naturally in the study of adversary methods
for a lifted query function.

To prove our product-to-sum reduction,
we employ a chain of reductions. First,
we show that if one of $\QIC(\Pi,\mu_0)$ or $\QIC(\Pi,\mu_1)$
is much smaller than the other, then we can use $\Pi$
to get a low-information protocol for $\OR\circ F$,
the composition of the $\OR$ function with $F$.
Next, we use an argument motivated by \cite{BBGK18}:
we use Belov's algorithm for the combinatorial
group testing problem \cite{Bel15} to use the low-information
cost protocol for $\OR\circ F$ to get a low-information
cost protocol for the task of computing $n$ copies of $F$.
Finally, we use an argument from \cite{Tou15} to get
a low-information cost protocol for $F$ itself.

\subsubsection*{Connection to secure two-party computation}

Another insight important for this work is that
lifting theorems for quantum adversary methods
are related to quantum secure two-party computation,
a cryptographic primitive. This connection
comes through the measure $\QICZ(G)$: for communication
gadgets $G$ for which $\QICZ(G)>0$, we know that
secure two-party computation of $G$ is impossible
(in an ``honest-but-curious'' setting, where we
require information-theoretic security);
yet for such $G$, we can then lift
$\Adv_1(f)$ to a lower bound on $\QCC(f\circ G)$.
We believe this result is likely to extend
to a lifting theorem for other quantum adversary methods
in the future -- though the dependence on $\QICZ(G)>0$
may still remain.

We provide a minimax theorem for $\QICZ(G)$, giving
an alternate characterization of the measure.
This minimax theorem is used in our lifting theorem,
and may also be useful for a future lower bound
on $\QICZ(G)$ for some communication function $G$,
which we view as an interesting open problem
arising from this work.

\subsubsection*{Insights into query complexity}

Our results for query complexity follow from the following
insights. First, we show that $\CAdv(f)$ is equivalent to
the measure $\cfbs(f)$ (a fractional version of critical block
sensitivity \cite{HN12}) by converting the primal
versions of the two programs to each other; this is not
difficult to do, and the main contribution comes from
(1) using the correct definition of $\CAdv(f)$
(out of the several definitions in \cite{AKPV18},
which are not equivalent to each other for partial functions),
and (2) using the correct definition of $\cfbs(f)$ (which
is a new definition introduced in this work).

Second, we show that the positive-weight adversary method
$\Adv(f)$ is smaller than, but quadratically related to,
$\CAdv(f)$. Once again, this result is not difficult,
but relies on using the correct definition of $\CAdv(f)$
and on using the primal versions (i.e.\ minimization versions)
of both programs. (Indeed, we use only the primal form
of all the adversary methods throughout this paper;
one of our insights is that this primal form
is more convenient for proving structural properties of the
adversary methods, including lifting theorems.)

Finally, we show that $\adeg(f)=\Omega(\sqrt{\cfbs(f)})$,
and hence $\adeg(f)=\Omega(\sqrt{\Adv(f)})$, and this holds
even for partial functions. We do this by essentially
reducing it to the task of showing $\adeg(f)=\Omega(\sqrt{\fbs(f)})$.
The latter is already known \cite{KT16}; however, it was
only known for total functions, whereas we need it to
hold for partial functions as well. The problem is that
the previous proof relied on recursively composing $f$
with itself, an operation which turns the fractional
block sensitivity $\fbs(f)$ into
the block sensitivity $\bs(f)$; unfortunately, this trick
works only for total Boolean functions. Instead,
we use a different trick for turning $\fbs(f)$
into $\bs(f)$: we compose $f$ with the promise-OR function,
and show that the block sensitivity of $f\circ\PrOR$
is proportional to the fractional block sensitivity of $f$.
We then convert an arbitrary polynomial approximating $f$
into a polynomial approximating $f\circ\PrOR$ by composing
it with a Chebyshev-like polynomial computing $\PrOR$;
finally, we appeal to the known result that the square root
of block sensitivity lower bounds approximate degree
to finish the proof.

\section{Preliminaries}
\subsection{Distance \& information measures}
We define all the distance and information measures for quantum states. The classical versions can be obtained by making the corresponding registers classical.

The $\ell_1$ distance between two quantum states $\rho$ and $\sigma$ is defined as
\[ \norm{\rho-\sigma}_1 = \Tr\sqrt{(\rho-\sigma)^\dagger(\rho-\sigma)}.\]

The entropy of a quantum state $\rho_A$ on register $A$ is defined as
\[H(A)_\rho = -\Tr(\rho\log\rho).\]
For a state $\rho_{AB}$ on registers $AB$, the conditional entropy of $A$ given $B$ is
\[ H(A|B)_\rho = H(AB)_\rho - H(B)_\rho.\]
Conditional entropy satisfies the following continuity bound \cite{AF04}: if $\rho$ and $\sigma$ on registers $AB$ satisfy $\norm{\rho-\sigma}_1 \leq \epsilon$, then
\[ |H(A|B)_\rho - H(A|B)_\sigma| \leq 4\epsilon\log|A| + 2h(\epsilon) \]
where $h(.)$ is the binary entropy function.
For $\rho_{ABC}$, we define the mutual information and conditional mutual information as
\[ I(A:B)_\rho = H(A)_\rho - H(A|B)_\rho \qquad I(A:B|C) = H(A|C)_\rho - H(A|BC)_\rho. \]
Mutual information satisfies
\[ 0 \leq I(A:B|C)_\rho \leq \min\{\log|A|,\log|B|\}\]
and the chain rule
\[ I(A:BC)_\rho = I(A:B)_\rho + I(A:C|B)_\rho.\]

\subsection{Query complexity}\label{sec:prelim_query}

In query complexity, the primary object of study
are Boolean functions, which are functions
$f\colon\B^n\to\B$ where $n$ is a positive integer. Often,
we will actually study partial Boolean functions, which
are defined on only a subset of $\B^n$. We will use
$\Dom(f)$ to denote the domain of $f$; this is a subset
of $\B^n$.

For a (possibly partial) Boolean function $f$, we use
$\D(f)$, $\R(f)$, and $\Q(f)$ to denote its deterministic
query complexity, randomized query complexity (to bounded error),
and quantum query complexity (to bounded error), respectively.
For the definition of these measures, see \cite{BdW02},
though we won't use these definitions in this work.

\subsubsection{Block sensitivity and its variants}

We will use the following definitions.

\paragraph{Block notation.} For a Boolean string $x\in\B^n$
and a set $B\subseteq[n]$, we let $x^B$ denote the
string with the bits in $B$ flipped; that is,
$x^B_i=x_i$ for all $i\notin B$ and $x^B_i=1-x_i$ for all $i\in B$.
The set $B$ is called a \emph{block}.

\paragraph{Sensitive block.} For a (possibly partial) Boolean
function $f$ on $n$ bits and an input $x\in\Dom(f)$, we say that
a set $B\subseteq[n]$ is a \emph{sensitive block} for $x$
(with respect to $f$) if $x^B\in\Dom(f)$ and $f(x^B)\ne f(x)$.

\paragraph{Block sensitivity.} The \emph{block sensitivity}
of a string $x\in\B^n$ with respect to a (possibly partial)
Boolean function $f$ satisfying $x\in\Dom(f)$ is the maximum
integer $k$ such that there are $k$ blocks
$B_1,B_2,\dots,B_k\subseteq[n]$ which are all sensitive
for $x$ and which are all disjoint. This is denoted
$\bs(x,f)$.

\paragraph{Block sensitivity of a function.} The
\emph{block sensitivity} of a (possibly partial) Boolean function
$f$ is the maximum value of $\bs(x,f)$ over $x\in\Dom(f)$.
This is denoted $\bs(f)$.
Block sensitivity was originally introduced by Nisan
\cite{Nis91}, and is discussed in the survey by Buhrman and
de Wolf \cite{BdW02}.

\paragraph{Fractional block sensitivity.} The
\emph{fractional block sensitivity} of a string $x\in\B^n$
with respect to a (possibly partial) Boolean function $f$
satisfying $x\in\Dom(f)$ is the maximum possible
sum of weights $\sum_B w_B$, where the weights $w_B\ge 0$
are assigned to each sensitive block of $x$ and must satisfy
$\sum_{B:i\in B} w_B\le 1$ for all $i\in[n]$.
This is denoted by $\fbs(x,f)$. The fractional block sensitivity
of a function $f$, denoted $\fbs(f)$, is the maximum value of
$\fbs(x,f)$ over $x\in\Dom(f)$.
Fractional block sensitivity was defined by \cite{Aar08},
but see also \cite{KT16}.

\paragraph{Critical block sensitivity.}
For a (possibly partial) Boolean function $f$, we say that
a total Boolean function $f'$ is a \emph{completion} of $f$
if $f'(x)=f(x)$ for all $x\in\Dom(f)$. The \emph{critical block
sensitivity} of $f$, denoted $\cbs(f)$, is defined as
\[\min_{f'}\max_{x\in\Dom(f)}\bs(x,f'),\]
where the minimum is taken over completions $f'$ of $f$.
This measure was defined by \cite{HN12}. It equals $\bs(f)$
for total functions, but may be larger for partial functions.

\paragraph{Critical fractional block sensitivity.}
For a (possibly partial) Boolean function $f$, we define
its \emph{critical fractional block sensitivity}, denoted
$\cfbs(f)$, as
\[\min_{f'}\max_{x\in\Dom(f)}\cfbs(x,f'),\]
where the minimum is taken over completions $f'$ of $f$.
This measure has not previously appeared in the literature.

\subsubsection{Adversary bounds}
\label{sec:prelim_adv}

\paragraph{Positive adversary bound.}
For a (possibly partial) Boolean function $f$,
we define the positive-weight adversary bound,
denoted $\Adv(f)$, as the minimum of the following program.
We will have one non-negative weight $q(x,i)$
for each $x\in\Dom(f)$ and each $i\in[n]$. We call
such a weight scheme feasible if, for all
$x,y\in\Dom(f)$ with $f(x)\ne f(y)$,
we have
\[\sum_{i:x_i\ne y_i} \sqrt{q(x,i)q(y,i)}\ge 1.\]
Then $\Adv(f)$ is defined as the minimum
of $\max_{x\in\Dom(f)}\sum_{i\in[n]} q(x,i)$
over feasible weight schemes $q(\cdot,\cdot)$.
A different version of the positive-weight
adversary bound was defined in \cite{Amb02},
though the version we've currently defined
appears in \cite{LM08} and \cite{SS06}
(in the latter, our definition is equivalent
to $\operatorname{MM}(f)$).

\paragraph{Classical adversary bound.}
For a (possibly partial) Boolean function $f$,
we define the classical adversary bound,
denoted $\CAdv(f)$, as the minimum of the following
program. We will have one non-negative
weight $q(x,i)$ for each $x\in\Dom(f)$ and each $i\in[n]$,
as before. We call such a weight scheme feasible
if, for all $x,y\in\Dom(f)$ with $f(x)\ne f(y)$,
we have
\[\sum_{i:x_i\ne y_i}\min\{q(x,i),q(y,i)\}\ge 1.\]
Then $\CAdv(f)$ is defined as the minimum
of $\max_{x\in\Dom(f)}\sum_{i\in[n]} q(x,i)$
over feasible weight schemes $q(\cdot,\cdot)$.
Observe that this definition is the same as
that of $\Adv(f)$, except that the feasibility
constraint sums up the minimum of $q(x,i)$ and
$q(y,i)$ instead of the geometric mean. This feasibility
constraint is harder to satisfy, and hence we have
$\CAdv(f)\ge\Adv(f)$.
A different version of the classical adversary
was defined in \cite{Aar06}, though the version
we've currently defined appears in \cite{LM08}
and \cite{AKPV18} (in the latter, our definition
is equivalent to $\operatorname{CMM}(f)$).

\paragraph{Singleton adversary bound.}
\cite{ABK+20} introduced a simplified version
of the quantum adversary bound, which we denote
$\Adv_1(f)$. As in the other adversaries,
this will be the minimum over a program that has
one non-negative weight $q(x,i)$ for each pair
of input $x\in\Dom(f)$ and index $i\in[n]$.
The objective value will once again be
$\max_{x\in\Dom(f)}\sum_{i\in[n]} q(x,i)$.
The only difference is the constraints:
instead of placing a constraint
for each $x,y\in\Dom(f)$ with $f(x)\ne f(y)$,
we only place this constraint for such $x,y$
that have Hamming distance exactly $1$.
Observe that this is a relaxation
of the constraint in the definition of $\Adv(f)$,
and hence $\Adv_1(f)\le \Adv(f)$ for all (possibly partial)
Boolean functions $f$.

\subsection{A generalization to relations}

So far, we've defined our query measures for partial
Boolean functions. However, in many cases we will be interested
in studying \emph{relations}, which are a generalization
of partial Boolean functions.

In query complexity, a \emph{relation} is a subset
of $\B^n\times \Sigma$, where $\Sigma$ is some finite
output alphabet. We will equate a relation
$f\subseteq \B^n\times\Sigma$ with a function
that maps $\B^n$ to subsets of $\Sigma$, so that
for a string $x\in\B^n$, the notation $f(x)$ denotes
$\{\sigma\in\Sigma:(x,\sigma)\in f\}$.
An algorithm which computes a relation $f$ to error
$\epsilon$ must have the guarantee that for inputs $x\in\B^n$,
the algorithm outputs a symbol in $f(x)$ with probability
at least $1-\epsilon$.

Relations are generalizations of partial functions.
This is because we can represent a partial function
$f$ with domain $\Dom(f)\subseteq\B^n$ by a relation
$f'$ such that $f'(x)=\{f(x)\}$ for $x\in\Dom(f)$
and $f'(x)=\{0,1\}$ for $x\notin\Dom(f)$. In other words,
the relational version $f'$ of the partial function $f$
will accept all input strings (it will be a total function),
but it will consider every output symbol to be valid
when given an input not in $\Dom(f)$. This essentially
makes the inputs not in $\Dom(f)$ become trivial, and hence
makes the relation $f'$ intuitively equivalent to the partial
function $f$.

We will generalize several of our query measures to relations.

\paragraph{Critical (fractional) block sensitivity.}
The original definition of $\cbs(f)$ from \cite{HN12} actually
defined it for relations.
We say that a total function $f'\colon\B^n\to\Sigma$
is a \emph{completion} of a relation $f\subseteq \B^n\times\Sigma$
if $(x,f'(x))\in f$ for all $x\in\B^n$. In other words,
$f'$ is a completion if it gives a fixed, valid
output choice for each input to $f$. Next, we say an input
$x\in\B^n$ is \emph{critical} if it has a unique valid
output symbol in $f$; that is, if $|f(x)|=1$.
We let $\crit(f)$ denote the set of all critical inputs
to $f$.
(Note that if $f$ is the relational version of a partial function,
then $\crit(f)$ is equal to the domain of the partial function.)
We then define
\[\cbs(f)\coloneqq \min_{f'}\max_{x\in\crit(f)}\bs(x,f')\]
\[\cfbs(f)\coloneqq \min_{f'}\max_{x\in\crit(f)}\fbs(x,f'),\]
where the minimizations are over completions $f'$ of $f$.
Observe that if $f$ is the relational version
of a partial function, these definitions
match the previous ones.

\paragraph{Adversary bounds.}
The adversary bounds easily generalize to relations:
both the positive adversary bound and the classical adversary
bound will still be minimizations over weight schemes
$q(x,i)$, with a non-negative weight assigned to each pair of
input in $\B^n$ and $i\in[n]$. The objective value to be
minimized is the same as before:
$\max_{x\in\B^n}\sum_{i\in[n]} q(x,i)$.
As for the constraints, we previously had one constraint
for each pair of inputs $x,y$ with $f(x)\ne f(y)$.
For relations, we will replace this condition with the condition
$f(x)\cap f(y)=\varnothing$ (that is, $x$ and $y$ have
disjoint allowed-output-symbol sets). Hence the new constraint
for $\Adv(f)$ becomes that for all pairs $x,y\in\B^n$
with $f(x)\cap f(y)=\varnothing$, we have
\[\sum_{i:x_i\ne y_i}\sqrt{q(x,i)q(y,i)}\ge 1.\]
Similarly, the constraint for $\CAdv(f)$ is that for all pairs
$x,y\in\B^n$ with $f(x)\cap f(y)=\varnothing$, we have
\[\sum_{i:x_i\ne y_i}\min\{q(x,i),q(y,i)\}\ge 1,\]
and the constraint for $\Adv_1(f)$ is similar.

\subsubsection{Degree measures}

\paragraph{Degree of a function.} For a
(possibly partial) Boolean function $f$, we define
its \emph{degree} to be the minimum degree of a real
polynomial $p$ which satisfies $p(x)=f(x)$
for all $x\in\Dom(f)$ as well as $p(x)\in[0,1]$ for
all $x\in\B^n$. We denote this by $\deg(f)$.

\paragraph{Approximate degree.} For a (possibly partial)
Boolean function $f$, we define its \emph{approximate degree}
to error $\epsilon$
to be the minimum degree of a real polynomial $p$
which satisfies $|p(x)-f(x)|\le \epsilon$
for all $x\in\Dom(f)$
as well as $p(x)\in[0,1]$ for all $x\in\B^n$.
We denote this by $\adeg_\epsilon(f)$.
When $\epsilon=1/3$, we omit it and write $\adeg(f)$.

These measures are both defined and discussed in the
survey by Buhrman and de Wolf \cite{BdW02}.
We note that for partial functions, some authors do not
include the requirement that the polynomial approximating
the function is bounded outside of the promise set.
Without this requirement, one gets a smaller measure.
In this work we will only use degree and approximate
degree to refer to the bounded versions of these measures.

We also note that approximate degree can be \emph{amplified}:
if a polynomial $p$ approximates a function $f$ to error
$\epsilon$, then we can modify $p$ to get a polynomial $q$
which approximates $f$ to error $\epsilon'<\epsilon$
and which has degree that is at most a constant factor
larger than the degree of $p$ (this constant
factor will depend on $\epsilon$ and $\epsilon'$).

\subsubsection{Known relationships between measures}

See \fig{relations} for a summary of relationships
between these measures (for partial functions).

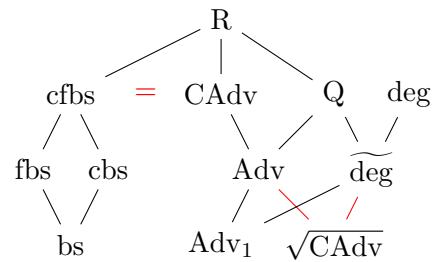
\begin{wrapfigure}{R}{0.35\textwidth}
\centering
  \begin{tikzpicture}[x=0.5cm,y=0.5cm]
     \node (cfbs) at(1,4){$\cfbs$};
     \node (fbs) at(0,2){$\fbs$};
     \node (cbs) at(2,2){$\cbs$};
     \node (bs) at(1,0){$\bs$};
     \node (eq) at(3,4){\color{red}$=$};
     \node (CAdv) at(5,4){$\CAdv$};
     \node (Adv) at(6,2){$\Adv$};
     \node (Advo) at(5,0){$\Adv_1$};
     \node (R) at(5,6){$\R$};
     
     \node (sqCAdv) at(8,0){$\sqrt{\CAdv}$};
     \node (adeg) at(9,2){$\adeg$};
     \node (Q) at(8,4){$\Q$};
     \node (deg) at(10,4){$\deg$};

     \path[-] (sqCAdv) edge[red] (Adv);
     \path[-] (sqCAdv) edge[red] (adeg);
     \path[-] (Adv) edge (Q);
     \path[-] (adeg) edge (Q);
     \path[-] (Q) edge (R);

     \path[-] (bs) edge (cbs);
     \path[-] (bs) edge (fbs);
     \path[-] (fbs) edge (cfbs);
     \path[-] (cbs) edge (cfbs);
     \path[-] (Advo) edge (Adv);
     \path[-] (Advo) edge (adeg);
     \path[-] (Adv) edge (CAdv);
     \path[-] (cfbs) edge (R);
     \path[-] (CAdv) edge (R);
     \path[-] (adeg) edge (deg);
  \end{tikzpicture}
    \caption{Relations between query complexity
    measures used in this work, applicable to partial functions.
    An upwards line from $\M_1(f)$ to $\M_2(f)$ means
    that $\M_1(f)=O(\M_2(f))$ for all (possibly partial) Boolean
    functions $f$. {\color{red}Red} indicates new
    relationships proved in this work.
    We warn that some of these relationships
    are false for relations; in particular, $\CAdv$ may
    be strictly larger than $\cfbs$ and its square root
    may be incomparable to $\adeg$ for relations.
    \label{fig:relations}}
\end{wrapfigure}

It is not hard to see that $\bs(f)$ is the smallest
of the block sensitivity measures, and $\cfbs(f)$ is
the largest. We know \cite{Aar08,HN12}
that $\fbs(f)$ and $\cbs(f)$ both lower bound $\R(f)$
for all (possibly partial) Boolean functions $f$;
in \sec{query}, we show that $\cfbs(f)$ is also
a lower bound.

We know \cite{LM08,SS06,AKPV18} that
$\Q(f)=\Omega(\Adv(f))$ and $\R(f)=\Omega(\CAdv(f))$.
Although this it not ordinarily stated for relations,
both lower bounds hold when $f$ is a relation as well.
In \sec{query}, we show that
$\CAdv(f)=\Theta(\cfbs(f))$
for all partial functions $f$,
and we also show that
$\CAdv(f)=O(\Adv(f)^2)$
which holds for both partial functions and relations.

Approximate degree lower bounds quantum query complexity:
$\Q(f)=\Omega(\adeg(f))$.
It is known \cite{BBC+01} that approximate
degree is lower-bounded by $\sqrt{\bs(f)}$.
Tal \cite{Tal13} showed that for total functions,
$\adeg(f)=\Omega(\sqrt{\fbs(f)})$.
In \sec{query}, we extend this result to partial
functions, and also prove that
$\adeg(f)=\Omega(\sqrt{\cfbs(f)})$.

In conclusion, $\CAdv(f)$ turns out to be the same
as $\cfbs(f)$ for partial functions,
and its square root lower bounds both
$\Adv(f)$ and $\adeg(f)$, both of which are lower bounds
on $\Q(f)$. Without taking square roots,
$\CAdv(f)$ is a lower bound on $\R(f)$
but not on $\Q(f)$.

When we move from partial functions
to relations, the measure $\CAdv(f)$ appears to get stronger
in comparison to the other measures: it is
strictly larger than $\cfbs(f)$, and appears to be
incomparable to $\adeg(f)$ (though defining the
latter for relations is a bit tricky, and we don't
do so in this work).

\subsection{Communication complexity}

In the communication model, two parties, Alice and Bob, are given inputs $x\in\X$ and $y\in\Y$ respectively, and in the most general case the task is to jointly compute a relation $f\subseteq \X\times\Y\times\Z$ by communicating with each other. In other words, on input $(x,y)$,
Alice and Bob must output a symbol
$z\in\Z$ such that $(x,y,z)\in f$.
Without loss of generality, we can assume Alice sends the first message, and Bob produces the output of the protocol.

In the classical randomized model, Alice and Bob are allowed to use shared randomness $R$ (and also possibly private randomness $R_A$ and $R_B$) in order to achieve this. The cost of a communication protocol $\Pi$, denoted by $\CC(\Pi)$ is the number of bits exchanged between Alice and Bob. The randomized communication complexity of a relation $f$ with error $\epsilon$, denoted by $\R^\CC_\epsilon(f)$, is defined as the minimum $\CC(\Pi)$ of a randomized protocol $\Pi$ that computes $f$ with error at most $\epsilon$ on every input.

\paragraph{Classical information complexity.} The information complexity of a protocol with inputs $X,Y$ according to $\mu$, shared randomness $R$ and transcript $\Pi$ is given by
\[ \IC(\Pi,\mu) = I(X:\Pi|YR)_\mu + I(Y:\Pi|XR)_\mu.\]
For any $\mu$ we have, $\IC(\Pi,\mu) \leq \CC(\Pi)$.

\paragraph{} In a quantum protocol $\Pi$, Alice and Bob initially share an entangled state on registers $A_0B_0$, and they get inputs $x$ and $y$ from a distribution $\mu$. The global state at the beginning of the protocol is
\[ \ket{\Psi^0} = \sum_{x,y}\sqrt{\mu(x,y)}\ket{xxyy}_{X\tX Y\tY}\otimes\ket{\Theta^0}_{A_0B_0}\]
where the registers $\tX$ and $\tY$ purify $X$ and $Y$ and are inaccessible to either party. In the $t$-th round of the protocol, if $t$ is odd, Alice applies a unitary $U_t: A'_{t-1}C_{t-1} \to A'_tC_t$, on her input, her memory register $A'_{t-1}$ and the message $C_{t-1}$ from Bob in the previous round (where $A'_0=XA_0$ and $C_0$ is empty), to generate the new message $C_t$, which she sends to Bob, and new memory register $A'_t$. Similarly, if $t$ is even, then Bob applies the unitary $U_t: B'_{t-1}C_{t-1} \to B'_tC_t$ and sends $C_t$ to Alice. It is easy to see that $B'_t=B'_{t-1}$ for odd $t$, and $A'_t=A'_{t-1}$ for even $t$. We can assume that the protocol is safe, i.e., for all $t$, $A'_t=XA_t$ and $B'_t = YB_t$, and $U_t$ uses $X$ or $Y$ only as control registers. The global state at the $t$-th round is then
\[ \ket{\Psi^t} = \sum_{x,y}\sqrt{\mu(x,y)}\ket{xxyy}_{X\tX Y\tY}\otimes\ket{\Theta^t}_{A_tB_tC_t|xy}.\]
\cite{LT17} (Proposition 9) showed that making a protocol safe does not decrease its QIC (defined below), so we shall often work with protocols of this form.

The quantum communication cost of a protocol $\Pi$, denoted by $\QCC(\Pi)$, is the total number of qubits exchanged between Alice and Bob in the protocol, i.e., $\sum_t\log|C_t|$. The quantum communication complexity of $f$ with error $\epsilon$, denoted by $\Q^\CC(f)$, is defined as the minimum $\QCC(\Pi)$ of a quantum protocol $\Pi$ that computes $f$ with error at most $\epsilon$ on every input.

\paragraph{Quantum information complexity.} Given a quantum protocol $\Pi$ as described above with classical inputs distributed as $\mu$, its quantum information complexity is defined as
\[ \QIC(\Pi,\mu) = \sum_{t \text{ odd}}I(\tX\tY:C_t|YB'_t)_{\Psi^t} + \sum_{t\text{ even}}I(\tX\tY:C_t|XA'_t)_{\Psi^t}.\]
The Holevo quantum information complexity is defined as
\begin{align*}
\HQIC(\Pi,\mu) & = \sum_{t\text{ odd}}I(X:B'_tC_t|Y)_{\Psi^t} + \sum_{t\text{ even}}I(Y:A'_tC_t|X)_{\Psi^t} \\
 & = \sum_{t\text{ odd}}I(X:B_tC_t|Y)_{\Psi^t} + \sum_{t\text{ even}}I(Y:A_tC_t|X)_{\Psi^t} \quad \text{(for safe protocols).}
\end{align*}
For brevity, we shall often only use the classical input distribution $\mu$ as the subscript, or drop the subscript entirely, for these information quantities.

It was proved in \cite{LT17}, that for an $r$-round protocol $\Pi^r$, HQIC and QIC satisfy the following relation:
\[ \QIC(\Pi^r, \mu) \geq \frac{1}{r}\HQIC(\Pi^r, \mu) \geq \frac{1}{2r}\QIC(\Pi^r,\mu).\]
Moreover, for any $\mu$, $\QIC(\Pi,\mu) \leq \QCC(\Pi)$.

\section{Lifting the classical adversary}

\subsection{The gadget and its properties}

The gadget we use is the same one used in \cite{GP18},
called \textsc{VER} in that work. This is the function
$\textsc{VER}\colon\B^2\times\B^2\to\B$ defined by $G(x,y)=1$
if and only if $x+y$ is equivalent to $2$ or $3$ modulo $4$,
where $x,y\in\B^2$ are interpreted as binary representations
of integers in $\{0,1,2,3\}$. This gadget has the property
of being \emph{versatile}, which means that it satisfies
the following three properties:
\begin{enumerate}
\item Flippability: given any input $(x,y)$, Alice
and Bob can perform a local operation on their respective inputs (without communicating)
to get $(x',y')$ such that $G(x',y')=1-G(x,y)$.
\item Random self-reducibility: given any input $(x,y)$,
Alice and Bob can use shared randomness (without communicating)
to generate $(x',y')$ which is uniformly distributed over
$G^{-1}(G(x,y))$. That is, Alice and Bob can convert any $0$-input
into a random $0$-input and any $1$-input into a random $1$-input,
without any communication. More formally,
if the domain of $G$ is $\X\times \Y$,
we require a probability distribution $\nu_G$
over pairs of permutations in $S_{\X}\times S_{\Y}$
such that for each $(x,y)\in\X\times\Y$,
sampling $(\sigma_A,\sigma_B)$ from $\nu_G$
and constructing the pair $(\sigma_A(x),\sigma_B(y))$
gives the uniform distribution over $G^{-1}(G(x,y))$.
\item Non-triviality: the function $G$ contains $\AND_2$
as a sub-function (and by flippability, it also contains
$\OR_2$ as a sub-function).
\end{enumerate}

These three properties were established in \cite{GP18} for
the function \textsc{VER}; a gadget which satisfies
them is called \emph{versatile}.
Our lifting proof will work for any versatile gadget $G$.
We will need the following simple lemma, which allows
us to generate $n$-tuples of inputs to $G$ that evaluate
to either a string $s\in\B^n$ or its complement
$\hat{s}\in\B^n$. We use the notation $G^{-1}(s)$
to denote the set of all $n$-tuples of inputs to $G$
that together evaluate to $s\in\B^n$; this abuses
notation slightly (we would technically need to write
$(G^{\oplus n})^{-1}(s)$, where $G^{\oplus n}$ is
the function we get by evaluating $n$ independent inputs
to $G$).

\begin{lemma}\label{lem:s-sampling}
Let $s\in\{0,1\}^m$ be a given string and $G$ be a versatile gadget. Then there is a protocol with no communication using shared randomness between Alice and Bob, who receive inputs $(a,b)$ in the domain of $G$ such that
\begin{itemize}
 \item If $G(a,b)=0$, Alice and Bob produce output strings $(x,y)$ that are uniformly distributed in $G^{-1}(s)$
 \item If $G(a,b)=1$, Alice and Bob produce output strings $(x,y)$ that are uniformly distributed in $G^{-1}(\hat{s}) = G^{-1}(s\oplus1^m)$.
\end{itemize}
\end{lemma}

\begin{proof}
Alice and Bob share independent instances of the permutations $\nu_G$, $\sigma_A$ and $\sigma_B$ as randomness. Applying independent instances of $\nu_G$, Alice and Bob can produce $(x',y')$ that are uniformly distributed in $G^{-1}((G(a,b))^m)$: this is done by applying $m$ independent instances of $\sigma_A$ and $\sigma_B$ from $\nu_G$ to $a$ and $b$ respectively. Now Alice and Bob know where $s$ differs from $0^m$. By applying independent instances of the local flipping operation on $x'$ and $y'$ at these locations, they can negate the output of $G$. It is clear the resultant string $(x,y)$ is uniformly distributed in $G^{-1}(s)$ if $G(a,b)=0$ and in $G^{-1}(\hat{s})$ if $G(a,b)=1$.
\end{proof}

We additionally have the following lemma, which
uses the non-triviality property of a versatile gadget.

\begin{lemma}\label{lem:ANDOR-classical}
If $G$ is a constant-sized non-trivial gadget (containing $\AND_2$ and $\OR_2$ as subfunctions), and $\mu_0$ and 
$\mu_1$ are uniform distributions over its $0$- and $1$-inputs, 
then any classical protocol $\Pi$ for computing $G$ with bounded 
error has $\IC(\Pi,\mu_0), \IC(\Pi,\mu_1) = \Omega(1)$.
\end{lemma}

\begin{proof}
$G$ contains the $\AND_2$ function, and $\mu_0$ puts uniform 
$\Omega(1)$ weight on the 0-inputs of the $\AND_2$ subfunction. 
\cite{BJKS04} showed that any protocol computing the $\AND_2$ 
function must have $\Omega(1)$ information cost with
respect to the distribution that puts
$1/3$ weight on all $0$-inputs of the $\AND_2$ function. Hence any 
protocol for $G$ must also have $\IC(\Pi,\mu_0)=\Omega(1)$. 
Similarly, by considering the fact that $G$ contains the $\OR_2$ 
function, we can show that $\IC(\Pi,\mu_1)=\Omega(1)$.
\end{proof}

Although we only need a single versatile gadget,
such as \textsc{VER}, we will briefly remark that there is
actually an infinite family of versatile gadgets, and that
this family is universal (i.e.\ every communication function
is a sub-function of some gadget in the family).

\begin{lemma}
There is a universal family of versatile gadgets.
\end{lemma}

\begin{proof}
For ease of notation, let $G$ denote \textsc{VER}.
For each $n\in\bN^+$, we define $G_{n}$ to be
$\PARITY_n\circ G$. Note that $G_{n}$ has the signature
$\B^{2n}\times\B^{2n}\to\B$. We observe that $G_{n}$
is versatile for each $n\in\bN^{+}$. This is because,
given a single input $((x_1,x_2,\dots,x_n),(y_1,y_2,\dots,y_n))$
to $G_{n}$ with $x_i,y_i\in\{0,1,2,3\}$ for each $i\in[n]$,
Alice and Bob can locally generate the uniform
distribution over all inputs with the same $G_{n}$-value.
They can do this by first negating a random subset of the
positions $i$ of even size
(using the flippability property of \textsc{VER}),
and then converting each of the $n$ resulting inputs to $G$
into a random input to $G$ with the same $G$-value.

Suppose $z$ is the $n$-bit string with $z_i=G(x_i,y_i)$.
Then flipping a random even subset of the bits of $z$
is equivalent to generating a random string $w$ that has
the same parity as $z$. It follows that the above procedure
generates a random input to $G_{n}$ that has the same $G_{n}$
value as the original input, meaning that $G_{n}$
is random self-reducible. By flipping any single gadget
$G$ within $G_{n}$, we can negate $G_n$, so it is also
flippable. Finally, since $G_n$ contains $G$ as a sub-function,
it also contains $\AND$ as a sub-function, so $G_n$
is versatile for each $n\in\bN^{+}$.

It remains to show that $\{G_n\}_n$ is universal.
We note that since $G$ contains $\AND$ as a sub-function,
and since $G_n=\PARITY_n\circ G$, the function $G_n$
contains $\PARITY_n\circ\AND$ as a sub-function. The latter
is the inner product function $\IP_n$ on $n$ bits,
which is well-known to be universal. Hence $G_n$ is also
universal.
\end{proof}

\subsection{The lifting theorem}

\begin{theorem}
Let $G$ be a constant-sized versatile gadget such as
$\textsc{VER}$, and let $f\colon\B^n\to\Sigma$ be a relation.
Then $\RCC(f\circ G) =\Omega(\CAdv(f))$.
\end{theorem}

\begin{proof}
Let $\Pi$ be a randomized protocol for $f\circ G$ which uses $T$ rounds
of communication (with one bit sent each round), and successfully computes $f\circ G$ with probability at least $1-\epsilon/2$ for each input. Consider inputs $XY$ distributed according to $\mu_z = \mu_{z_1}\otimes\ldots\otimes\mu_{z_n}$, where each $\mu_{z_i}$ is the uniform distribution over $(x_i,y_i)$ in $G^{-1}(z_i)$. Suppose $\Pi$ uses public randomness $R$ which is independent of the inputs $XY$. We introduce the dependency-breaking random variables $D$ and $U$ \cite{BJKS04} in the following way: $D$ is independent of $X, Y, R$ and is uniformly distributed on $\B^n$. For each $i \in [n]$, if $D_i=0$, then $U_i = X_i$, and if $D_i=1$, then $U_i=Y_i$. Defined this way, given $D_iU_i$, $X_i$ and $Y_i$ are independent under $\mu_z$. We shall use this algorithm to give a weight scheme $q'(z,i)$:
\[q'(z,i)=I(X_i:\Pi|X_{<i} YDUR)_{\mu_z} + I(X_i:\Pi|Y_{<i}XDUR)_{\mu_z}\]
where $X_{<i}$ denotes $X_1\ldots X_{i-1}$, and similarly for $Y_{<i}$. Clearly $q$ is non-negative, and we shall show that
\begin{align*}
& \sum_{i:z_i\neq w_i}\min\{q'(z,i),q'(w,i)\} = \Omega(1)
\end{align*}
for all $z,w$ such that $f(z) \cap f(w) = \varnothing$, where the constant in the $\Omega(1)$ is universal. Using this constant to normalize $q'(z,i)$, we get $q(z,i)$ which is a valid weight scheme. Since for any fixed value of $DU=du$, $I(X:\Pi|YR)_{\mu_{zdu}}$ is an information cost,
\[ \sum_{i \in [n]}q'(z,i) = I(X:\Pi|YDUR)_{\mu_z} + I(Y:\Pi|XDUR)_{\mu_z} \leq \CC(\Pi),\]
we have for any protocol $\Pi$,
\[ \CC(\Pi) \geq \Omega\left(\sum_{i \in [n]}q(z,i)\right) \geq \Omega\left(\min_{\{q(z,i)\}}\sum_{i \in [n]}q(z,i)\right) \]
where the minimization is over all valid weight schemes. This proves the result.

Let $z$ and $w$ be two inputs to $f$ such that $f(z) \cap f(w) =\varnothing$. Suppose $z$ and $w$ differ on indices in the block $\cB$. Let $\cB^1$ be the subset of indices in $\cB$ where $\min\{q'(z,i),q'(w,i)\}$ is achieved by $q'(z,i)$, and $\cB^2$ be the subset where the minimum is achieved by $q'(w,i)$. For an index $i\in \cB^1$, let $\cB^1_i$ denote $\cB^1 \cap [i-1]$, and $\cB^2_i$ denote $\cB^2 \cap [i-1]$. We also use $\cB^{1,c}$ to denote $[n]\setminus\cB^1$, and $\cB^{1,c}_i$ to denote $[i-1]\setminus\cB^1_i$. Then,
\begin{align}
&\sum_{i: z_i \neq w_i}\min\left\{q'(z,i), q'(w,i)\right\} \nonumber \\
& = \sum_{i\in \cB^1}(I(X_i:\Pi|X_{<i}YDUR)_{\mu_z} + I(Y_i:\Pi|Y_{<i}XDUR)_{\mu_z}) \nonumber \\
& \quad + \sum_{i\in \cB^2}(I(X_i:\Pi|X_{<i}YDUR)_{\mu_w} + I(Y_i:\Pi|Y_{<i}XDUR)_{\mu_w}) \nonumber \\
& \overset{(1)}= \frac{1}{2}\sum_{i\in \cB^1}(I(X_i:\Pi|X_{<i}YD_{-i}U_{-i}R)_{\mu_z} + I(Y_i:\Pi|Y_{<i}XD_{-i}U_{-i}R)_{\mu_z}) \nonumber \\
& \quad + \frac{1}{2}\sum_{i\in \cB^2}(I(X_i:\Pi|X_{<i}YD_{-i}U_{-i}R)_{\mu_w} + I(Y_i:\Pi|Y_{<i}XD_{-i}U_{-i}R)_{\mu_w}) \nonumber \\
& \overset{(2)}\geq \frac{1}{2}\sum_{i\in \cB^1}(I(X_i:\Pi|X_{\cB^1_i}Y_{\cB^1}D_{\cB^{1,c}}U_{\cB^{1,c}}R)_{\mu_z} + I(Y_i:\Pi|Y_{\cB^1_i}X_{\cB^1}D_{\cB^{1,c}}U_{\cB^{1,c}}R)_{\mu_z}) \nonumber \\
& \quad + \frac{1}{2}\sum_{i\in \cB^2}(I(X_i:\Pi|X_{\cB^2_i}Y_{\cB^2}D_{\cB^{2,c}}U_{\cB^{2,c}}R)_{\mu_w} + I(Y_i:\Pi|Y_{\cB^2_i}X_{\cB^2}D_{\cB^{2,c}}U_{\cB^{2,c}}R)_{\mu_w}) \nonumber \\
& = \frac{1}{2}(I(X_{\cB^1}:\Pi|Y_{\cB^1}D_{\cB^{1,c}}U_{\cB^{1,c}}R)_{\mu_z} + I(Y_{\cB^1}:\Pi|X_{\cB^1}D_{\cB^{1,c}}U_{\cB^{2,c}}R)_{\mu_z}) \nonumber \\
& \quad + \frac{1}{2}(I(X_{\cB^2}:\Pi|Y_{\cB^2}D_{\cB^{2,c}}U_{\cB^{2,c}}R)_{\mu_w} + I(Y_{\cB^2}:\Pi|X_{\cB^2}D_{\cB^{2,c}}U_{\cB^{2,c}}R)_{\mu_w}). \label{inf_ineq}
\end{align}
Above, equality $(1)$ follows by using the definition of $D_iU_i$. The inequality $(2)$ follows from the fact that given $Y_i$, $X_i$ is independent of all other $X_j$-s, $Y_j$-s $D_j$-s and $U_j$-s under both the $z$ and $w$ distributions, hence $I(Y_i:\Pi|Y_{<i}X(DU)_{-i}R)_{\mu_z} \geq I(Y_i:\Pi|Y_{\cB^1_i}X_{\cB^1}(DU)_{\cB^{1,c}}R)_{\mu_z}$, and equivalent inequalities hold for the other terms.

Consider $v \in \{0,1\}^n$ which agrees with $w$ on the bits in $\cB^1$, with $z$ on the bits in $\cB^2$, and with both of them outside $\cB$. Since $f(z)$ and $f(w)$ are disjoint, at least one of the following must be true:
\begin{enumerate}
    \item $\Pr_{(x,y)\sim \mu_v}[\Pi(x,y) \in f(z)] \leq \frac{1}{2}$
    \item $\Pr_{(x,y)\sim \mu_v}[\Pi(x,y) \in f(w)] \leq \frac{1}{2}$.
\end{enumerate}
In case 1, we shall give a protocol $\Pi'$ that computes $G$ correctly with probability at least $1-\epsilon$ in the worst case, such that
\[\IC(\Pi', \mu_0) = O( I(X_{\cB^1}:\Pi|Y_{\cB^1}D_{\cB^{1,c}}U_{\cB^{1,c}}R)_{\mu_z} + I(Y_{\cB^1}:\Pi|X_{\cB^1}D_{\cB^{1,c}}U_{\cB^{1,c}}R)_{\mu_z}).\]
Similarly, in case 2, we can use $\Pi$ to give a protocol $\Pi''$ for $G$, such that
\[ \IC(\Pi'', \mu_1) = O(I(X_{\cB^2}:\Pi|Y_{\cB^2}D_{\cB^{2,c}}U_{\cB^{2,c}}R)_{\mu_w} + I(Y_{\cB^2}:\Pi|X_{\cB^2}D_{\cB^{2,c}}U_{\cB^{2,c}}R)_{\mu_w}).\]
Due to equation \eqref{inf_ineq} and \lem{ANDOR-classical}, this proves the theorem.

In fact we only show how to construct the protocol $\Pi'$ in case 1; the construction of $\Pi''$ is identical. Since $z$ is in the domain of $f$ and $\Pi$ has worst case correctness for $f\circ G$, we must have $\Pr_{(x,y)\sim\mu_z}[\Pi(x,y) \in f(z)] \geq 1-\epsilon/2$. Therefore, in case 1, $\Pi$ can distinguish between samples from $\mu_z$ and $\mu_v$ on average: on getting a sample from $\mu_z$ or $\mu_v$, we can run $\Pi$ to see if it gives an output in $f(z)$ or not, and output $z$ or $v$ accordingly. This average distinguishing probability can be boosted by running $\Pi$ multiple times. 

In $\Pi'$, Alice and Bob will share $RD_{\cB^{1,c}}U_{\cB^{1,c}}R_AR_B$ as randomness, where we use $R_A$ and $R_B$ to denote Alice and Bob's part of the shared randomness from  \lem{s-sampling}, required to generate $z_{\cB^1}$ if $G(a,b)=0$ and $v_{\cB^1}$ if $G(a,b)=1$. On input $(a,b)$ to $G$, Alice and Bob do the following $k$ times for $k=\frac{2}{\epsilon^2}\ln(1/\epsilon)$ :
\begin{itemize}
\item Alice sets $x_{\cB^1} = R_A(a)$ and Bob sets $y_{\cB^1} = R_B(b)$.
\item Alice samples $x_{\cB^{1,c}}$ and Bob samples $y_{\cB^{1,c}}$ from private randomness, so that $G^{|\cB^{1,c}|}(x_{\cB^{1,c}},y_{\cB^{1,c}})$ $= z_{\cB^{1,c}}$. They can do this since given $(DU)_{\cB^{1,c}}$, $X_{\cB^{1,c}}$ and $Y_{\cB^{1,c}}$ are independent under $\mu_z$ and $\mu_v$.
\item They run $\Pi$ on this $(x,y)$ and generate the corresponding output.
\end{itemize}
(There are $k$ independent instances of the $D, U, R_A, R_B$ variables for each run above, but we denote all of them the same way for brevity.) The final output of $\Pi'$ is 1 if the number of runs which have given an output in $f(z)$ is at least $(1-\epsilon)k$, and 0 otherwise.

Clearly if $G(a,b)=0$, then $(x,y)$ generated this way is uniformly distributed in the support of $\mu_z$, and if $G(a,b)=1$, then $(x,y)$ is uniform in the support of $\mu_v$. Calling the protocol in the $i$-th round of $\Pi'$, $\Pi_i$, notice that the transcript of each $\Pi_i$ is independent of $AR_A$, where $A$ is the random variable for Alice's input, given the generated $X_{\cB^1}$ (and this holds true even conditioned on $BR_BD_{\cB^{1,c}}U_{\cB^{1,c}}R$). Moreover, both $X_{\cB^1}$ and $\Pi_i$ are independent of $BR_B$ given $Y_{\cB^1}$ and of $Y_{\cB^1}$ given $BR_B$ (even conditioned on $D_{\cB^1}U_{\cB^1}R$). Let $\mu_{0,z}$ denote the distribution of $ABR_AR_B(DU)_{\cB^{1,c}}$ when $G(a,b)=0$, which induces the distribution $\mu_{z_{\cB^1}}$ on $X_{\cB^1}Y_{\cB^1}$. Hence,
\begin{align*}
I(A:\Pi_i|BR_AR_B(DU)_{\cB^1}R)_{\mu_{0,z}} & \leq I(AR_A:\Pi_i|BR_B(DU)_{\cB^{1,c}}R)_{\mu_{0,z}} \\
 & \overset{(1)}{\leq} I(X_{\cB^1}:\Pi_i|BR_B(DU)_{\cB^{1,c}}R)_{\mu_{0,z}} \\
 & \overset{(2)}{\leq} I(X_{\cB^1}:\Pi_i|Y_{\cB^1}BR_B(DU)_{\cB^{1,c}}R)_{\mu_{0,z}} \\
 & \overset{(3)}{=} I(X_{\cB^1}:\Pi_i|Y_{\cB^1}(DU)_{\cB^{1,c}}R)_{\mu_{0,z}} \\
 & \overset{(4)}{=} I(X_{\cB^1}:\Pi|Y_{\cB^1}(DU)_{\cB^{1,c}}R)_{\mu_z}.
\end{align*}
where inequality (1) follows from the fact that $I(U':V|W) \leq I(U:V|W)$ if $U'$ is independent of $V$ given $UW$, (2) follows from the fact that $I(U:V|W) \leq I(U:V|WW')$ if $V$ is independent of $W'$ given $W$, and (3) follows from $I(U:V|WW')=I(U:V|W')$ if $U$ and $V$ are both independent of $W$ given $W'$. Finally, (4) follows from the definition of $\Pi_i$ and the the fact that the variables $AB$ don't appear in the expression, so we can switch from $\mu_{0,z}$ to $\mu_z$. Similarly,
\[ I(B:\Pi_i|AR_AR_B(DU)_{\cB^{1,c}}R)_{\mu_{0,z_{\cB^1}}} \leq I(Y_{\cB^1}:\Pi|X_{\cB^1}(DU)_{\cB^{1,c}}R)_{\mu_{z_{\cB^1}}},\]
which lets us conclude that
\[\IC(\Pi_i,\mu_0) \leq I(X_{\cB^1}:\Pi|Y_{\cB^1}(DU)_{\cB^{1,c}}R)_{\mu_z} + I(Y_{\cB^1}:\Pi|X_{\cB^1}(DU)_{\cB^{1,c}}R)_{\mu_z}.\]
The final $\IC(\Pi',\mu_0)$ is then at most $k(I(X_{\cB^1}:\Pi|Y_{\cB^1}(DU)_{\cB^{1,c}}R)_{\mu_z} + I(Y_{\cB^1}:\Pi|X_{\cB^1}(DU)_{\cB^{1,c}}R)_{\mu_z})$.

Now let us analyze the worst case error made by $\Pi'$. Since the output of $\Pi_i$ on $(a,b)$ is expected output of $\Pi$ on $(x,y)$ uniformly sampled from either $\mu_z$ or $\mu_v$, $\Pi_i$ produces an output in $f(z)$ on $(a,b)$ such that $G(a,b)=0$ with probability at least $1-\epsilon/2$, and on $(a,b)$ such that $G(a,b)=1$ with probability at most $\frac{1}{2}$. Hence by the Hoeffding bound, the probability of $(1-\epsilon)k$ many 0 outputs in the first case is at least
\[ 1-e^{-\epsilon^2k/2} \geq 1- \epsilon\]
and in the second case is at most
\[ e^{-2(1/2-\epsilon)^2k} \leq \epsilon.\]
Hence the probability of error on either input is at most $\epsilon$.
\end{proof}

\section{Quantum bounded-round lifting}

The following result, analogous to \lem{ANDOR-classical} except with round dependence, holds in the quantum case.

\begin{lemma}\label{lem:ANDOR-quantum}
Let $G$ be a constant-sized gadget which contains
$\AND_2$ and $\OR_2$ as sub-functions, and $\mu_0$ and $\mu_1$ be uniform distributions over its $0$- and $1$-inputs.
Then any $r$-round quantum protocol $\Pi$ for computing $G$
with bounded error has
$\QIC(\Pi,\mu_0), \QIC(\Pi,\mu_1) = \tOmega(1/r)$.
\end{lemma}

The lemma has a similar proof to the classical case, and invokes the near-optimal lower bound for the quantum information
cost of the $\AND_2$ and $\OR_2$ functions due to \cite{BGK+18}.

\begin{theorem}
If $G$ is a constant-sized versatile gadget,
then $\QCC^r(f\circ G) = \tOmega(\CAdv(f)/r^2)$.
\end{theorem}

\begin{proof}
For an $r$-round quantum protocol $\Pi$ that computes $f\circ G$ to error at most $\epsilon/2$, we define
\begin{align*}
q'(z,i) & = \sum_{t\text{ odd}}I(X_i:B_tC_t|X_{<i}YDU)_{\mu_z} + \sum_{t\text{ even}}I(Y_i:A_tC_t|Y_{<i}XDU)_{\mu_z}
\end{align*}
where the the distribution $\mu_z$ and correlation-breaking variables $DU$ are as in the classical case. Clearly,
\begin{align*}
\frac{1}{r}\sum_{i=1}^n q'(z,i) & = \frac{1}{r} \sum_{t\text{ odd}}I(X:B_tC_t|YDU)_{\mu_z} + \sum_{t\text{ even}}I(Y:A_tC_t|XDU)_{\mu_z} \\
& = \frac{1}{r}\HQIC(\Pi,\mu_z) \leq \QIC(\Pi,\mu_z) \leq \QCC(\Pi).
\end{align*}
Clearly $q'(z,i)$ is non-negative, and for all $z,w$ such that $f(z)\cap f(w) = \varnothing$, we shall show that
\begin{equation}\label{eq:QIC-zw} \sum_{i:z_i\neq w_i}\min\{q'(z,i),q'(w,i)\} = \widetilde{\Omega}(1/r).\end{equation}
Thus, defining $q(z,i)$ as our weight scheme by normalizing $q'(z,i)$ with the $r$ factor, we get the required result.

Showing \eqref{eq:QIC-zw} proceeds very similar to the classical case. For two inputs $z, w$ to $f$ such that $f(z) \cap f(w) = \varnothing$, which differ on the bits in block $\cB$, let $\cB^1 \subseteq \cB$ be the indices where $\min\{q'(z,i), q'(w,i)\}$ is achieved by $q'(z,i)$, and $\cB^2 \subseteq \cB$ be the indices where it is achieved by $q'(w,i)$. By the same chain of inequalities as in the classical case, we have
\begin{align*}
& \sum_{i:z_i\neq w_i}\min\{q'(z,i), q'(w,i)\} \\
& \geq \frac{1}{2}\left(\sum_{t\text{ odd}}I(X_{\cB^1}:B_tC_t|Y_{\cB^1}D_{\cB^{1,c}}U_{\cB^{1,c}})_{\mu_z} + \sum_{t\text{ even}}I(Y_{\cB^1}:A_tC_t|X_{\cB^1}D_{\cB^{1,c}}U_{\cB^{2,c}})_{\mu_z}\right) \\
& \quad + \frac{1}{2}\left(\sum_{t\text{ odd}}I(X_{\cB^2}:B_tC_t|Y_{\cB^2}D_{\cB^{2,c}}U_{\cB^{2,c}})_{\mu_w} + \sum_{t\text{ even}}I(Y_{\cB^2}:A_tC_t|X_{\cB^2}D_{\cB^{2,c}}U_{\cB^{2,c}})_{\mu_w}\right).
\end{align*}
Note that if we had used a QIC-based definition, instead of an HQIC-based definition, for $q'(z,i)$, where we conditioned on the $B_t, A_t$ registers, the above chain of inequalities would not have been valid, since $X_i$ is not independent of $X_jY_jD_jU_j$ at $j\neq i$ conditioned on $B_t$, and the same holds for $Y_i$.

Define the hybrid input $v$ which agrees with $w$ on the bits in $\cB^1$, with $z$ on the bits in $\cB^2$ and with both outside $\cB$. At least one of the following is true of $v$:
\begin{enumerate}
    \item $\Pr_{(x,y)\sim \mu_v}[\Pi(x,y) \in f(z)] \leq \frac{1}{2}$
    \item $\Pr_{(x,y)\sim \mu_v}[\Pi(x,y) \in f(w)] \leq \frac{1}{2}$.
\end{enumerate}
In case 1, we shall give a protocol $\Pi'$ that computes $G$ correctly with probability at least $1-\epsilon$ in the worst case, such that
\[\HQIC(\Pi', \mu_0) = O\left(\sum_{t\text{ odd}}I(X_{\cB^1}:B_tC_t|Y_{\cB^1}D_{\cB^{1,c}}U_{\cB^{1,c}})_{\mu_z} + \sum_{t\text{ even}}I(Y_{\cB^1}:A_tC_t|X_{\cB^1}D_{\cB^{1,c}}U_{\cB^{1,c}})_{\mu_z}\right).\]
Similarly, in case 2, we can use $\Pi$ to give a protocol $\Pi''$ for $G$, such that
\[ \HQIC(\Pi'', \mu_1) = O\left(\sum_{t\text{ odd}}I(X_{\cB^2}:B_tC_t|Y_{\cB^2}D_{\cB^{2,c}}U_{\cB^{2,c}})_{\mu_w} + \sum_{t\text{ even}}I(Y_{\cB^2}:A_tC_t|X_{\cB^2}D_{\cB^{2,c}}U_{\cB^{2,c}})_{\mu_w}\right).\]
The number of rounds in $\Pi'$ and $\Pi''$ will be $kr$, for $k=\frac{2}{\epsilon^2}\ln(1/\epsilon)$. This proves the theorem due to Lemma \ref{lem:ANDOR-quantum}, and the fact that $\HQIC(\Pi', \mu) =\Omega( \QIC(\Pi',\mu))$ for any $\mu$.

We only describe the protocol $\Pi'$. In $\Pi'$, Alice and Bob will share the initial entangled state of $\Pi$, as well as $D_{\cB^{1,c}}U_{\cB^{1,c}}R_AR_B$ as randomness, where $R_A$ and $R_B$ are Alice and Bob's parts of the shared randomness from Lemma \ref{lem:s-sampling}. Note that sharing randomness is equivalent to sharing an entangled state whose Schmidt coefficients are equal to the square roots of the corresponding probabilities, and locally measuring this state to get classical variables to use. We denote the inputs of $\Pi'$ by $(x',y')$ here to avoid confusion with the memory registers. On input $(x',y')$, Alice and Bob do the following $k$ times in $\Pi'$:
\begin{itemize}
\item Alice sets $x_{\cB^1} = R_A(x')$ and Bob sets $y_{\cB^1} = R_B(y')$.
\item Alice samples $x_{\cB^{1,c}}$ and Bob samples $y_{\cB^{1,c}}$ from private randomness (this can be done unitarily), so that $G^{|\cB^{1,c}|}(x_{\cB^{1,c}},y_{\cB^{1,c}})$ $= z_{\cB^{1,c}}$. They can do this since given $(DU)_{\cB^{1,c}}$, $X_{\cB^{1,c}}$ and $Y_{\cB^{1,c}}$ are independent under $\mu_z$ and $\mu_v$.
\item They run $\Pi$ on this $(x,y)$ and generate the corresponding output.
\end{itemize}
The final output of $\Pi'$ is 1 if the number of runs which have given an output in $f(z)$ is at least $(1-\epsilon)k$, and 0 otherwise.

Let $\mu_{0,z}$ denote the distribution of $X'Y'R_AR_B(DU)_{\cB^{1,c}}$ when $G(x',y')=0$, which induces $\mu_{z_{\cB^1}}$ on $X_{\cB^1}Y_{\cB^1}$. Let $C_{t,i}$ denote the message and $A_{t,i}, B_{t,i}$ the memory registers of the $i$-th run of $\Pi$ in $\Pi'$, which we denote by $\Pi_i$. (There are also independent $D, U, R_A, R_B$ variables for each run, but we drop the $i$ dependence here.) For every $i$, and an odd round $t$, we have similar to the classical case,
\[ I(X':B_{t,i}C_{t,i}|Y'R_AR_B(DU)_{\cB^{1,c}})_{\mu_{0,z}} \leq I(X_{\cB^1}:B_tC_t|Y_{\cB^1}(DU)_{\cB^{1,c}})_{\mu_z} \]
Similarly, for even $t$,
\[ I(Y':A_{t,i}C_{t,i}|X'R_AR_B(DU)_{\cB^{1,c}})_{\mu_{0,z}} \leq I(Y_{\cB^1}:A_tC_t|X_{\cB^1}(DU)_{\cB^{1,c}})_{\mu_z} \]
which gives us
\[ \HQIC(\Pi_i, \mu_0) \leq \sum_{t\text{ odd}}I(X_{\cB^1}:B_tC_t|Y_{\cB^1}(DU)_{\cB^{1,c}})_{\mu_z} + \sum_{t\text{ even}}I(Y_{\cB^1}:A_tC_t|X_{\cB^1}(DU)_{\cB^{1,c}})_{\mu_z}.\]
Finally, $\HQIC(\Pi', \mu_0) = k\HQIC(\Pi_i,\mu_0)$.

Since $z$ is in the domain of $f$ and $\Pi$ is correct for $f\circ G$ with probability at least $1-\epsilon/2$, we have $\Pr_{(x,y) \sim \mu_z}[\Pi(x,y) \in f(z)] \geq 1-\epsilon/2$, and the probability when $(x,y)$ is sampled according to $\mu_v$ instead is at most $\frac{1}{2}$. Therefore, by the definition of $\Pi'$ and the Hoeffding bound, $\Pi'$ is correct for $G$ with probability at least $1-\epsilon$. This completes the proof.
\end{proof}

\section{Towards a full quantum adversary lifting theorem}
\label{sec:Adv_one}

In this section, we will prove a conditional
lifting theorem for a somewhat weak
quantum adversary method, $\Adv_1$. The goal
of this section is primarily to introduce some tools
that we believe will be helpful in eventually
proving a lifting theorem for the positive-weight
quantum adversary method (hopefully
with a constant-sized gadget such as the \textsc{VER}).
Specifically, we prove a product-to-sum reduction
for quantum information cost in \sec{am_gm},
which should be helpful for handling the
$\sqrt{q(z,i)q(w,i)}$ terms that occur in the positive-weight
adversary method; indeed, we use this product-to-sum
reduction for our $\Adv_1$ lifting theorem.
We also show how lifting theorems for quantum adversary
methods are related to 2-party secure communication.

We now introduce the definition of $\QICZ(G)$,
our measure of the information leak that must
happen in any purported 2-party secure computation of $G$.

\begin{definition}
Let $G\colon\X\times\Y\to\B$ be a communication function.
Let $P$ be the set of all communication protocols
which solve $G$ to worst-case error $1/3$.
Let $\Delta_0$ be the set of all probability
distributions over $G^{-1}(0)$, and let
$\Delta_1$ be the set of all probability distributions
over $G^{-1}(1)$. We define
\[\QICZ(G)\coloneqq\adjustlimits
\inf_{\Pi\in P}\sup_{\mu\in\Delta_0\cup\Delta_1}
\QIC(\Pi,\mu).\]
\end{definition}

We note that since $\QIC(\Pi,\cdot)$ is a continuous
function of distributions \cite{BGK+18}, the inner
supremum is actually attained as a maximum.
We can now state our lifting theorem, as follows.

\begin{restatable}{theorem}{Advliftformal}\label{thm:Adv_lift_formal}
Let $f\colon\B^n\to\Sigma$ be a relation
(where $n\in\bN^{+}$ and $\Sigma$ is a finite alphabet)
and let $G\colon\X\times\Y\to\B$ be a communication
function which contains both $\AND_2$ and $\OR_2$
as subfunctions. Then
\[\QCC(f\circ G)=\tOmega(\Adv_1(f)\QICZ(G)).\]
\end{restatable}

\subsection{A minimax for \texorpdfstring{$\QICZ$}{QICZ}}

Before attacking the proof of \thm{Adv_lift_formal},
we first prove a minimax theorem for the measure
$\QICZ(G)$, giving an alternate characterization of it.
To do so, we invoke Sion's minimax theorem \cite{Sio58}.

\begin{fact}[Sion's minimax]\label{fct:Sion}
Let $V_1$ and $V_2$ be real topological vector spaces,
and let $X\subseteq V_1$ and $Y\subseteq V_2$ be convex.
Let $\alpha\colon X\times Y\to\bR$ be semicontinuous and
saddle. If either $X$ or $Y$ is compact, then
\[\adjustlimits\inf_{x\in X}\sup_{y\in Y}\alpha(x,y)
=\adjustlimits\sup_{y\in Y}\inf_{x\in X}\alpha(x,y).\]
\end{fact}

To understand the statement of this theorem,
we need a few definitions:

\begin{enumerate}
\item A real-valued function $\phi$ is \emph{convex}
if $\phi(\lambda x_1+(1-\lambda)x_2)\le
\lambda \phi(x_1)+(1-\lambda)\phi(x_2)$
for all $x_1,x_2\in\Dom(\phi)$ and all $\lambda\in(0,1)$.
\item A real-valued function $\phi$ is \emph{concave}
if $\phi(\lambda x_1+(1-\lambda)x_2)\ge
\lambda \phi(x_1)+(1-\lambda)\phi(x_2)$
for all $x_1,x_2\in\Dom(\phi)$ and all $\lambda\in(0,1)$.
\item A function $\alpha\colon X\times Y\to\bR$ is \emph{saddle}
if $\alpha(\cdot, y)$ is convex as a function of
$x$ for each fixed $y\in Y$, and if $\alpha(x,\cdot)$
is concave as a function of $y$ for each fixed $x\in X$.
\item A real-valued function $\phi$ is \emph{upper semicontinuous}
at a point $x$ if for any $\epsilon>0$, there exists
a neighborhood $U$ of $x$ such that for all $x'\in U$,
we have $\phi(x')<\phi(x)+\epsilon$.
\item A real-valued function $\phi$ is \emph{lower semicontinuous}
at a point $x$ if for any $\epsilon>0$, there exists
a neighborhood $U$ of $x$ such that for all $x'\in U$,
we have $\phi(x')>\phi(x)-\epsilon$.
\item A function $\alpha\colon X\times Y\to\bR$
is \emph{semicontinuous} if $\alpha(\cdot,y)$ is lower semicontinuous
over all of $X$ for each $y\in Y$ and if $\alpha(x,\cdot)$
is upper semicontinuous over all of $Y$ for each $x\in X$.
\end{enumerate}

We now use Sion's minimax theorem to prove a minimax theorem
for $\QICZ$.

\begin{theorem}\label{thm:QICZ_minimax}
Fix a communication function $G$. Let $P$
be the set of all protocols which solve $G$ to worst-case
error $1/3$, let $\Delta_0$ be the set of probability
distributions over $0$-inputs to $G$, and let
$\Delta_1$ be the set of probability distributions
over $1$-inputs to $G$. Then
\[\frac12 
\max_{\substack{\mu_0\in\Delta_0\\\mu_1\in\Delta_1}}
\inf_{\Pi\in P}\QIC(\Pi,\mu_0)+\QIC(\Pi,\mu_1)
\le \QICZ(G)\le 
\max_{\substack{\mu_0\in\Delta_0\\\mu_1\in\Delta_1}}
\inf_{\Pi\in P}\QIC(\Pi,\mu_0)+\QIC(\Pi,\mu_1).\]
Moreover, the maximum is attained.
\end{theorem}

\begin{proof}
We will aim to use Sion's minimax theorem \cite{Sio58}.
To this end, we start with a bit of formalism.
The set $P$ of protocols is, of course, an infinite set,
and has somewhat complicated structure.
In order to apply a minimax theorem, however, we want
to switch over to a convex subset of a real topological
vector space. To do so, we first consider the
free real vector space over $P$, which we denote
by $V(P)$. This is the real vector space consisting
of all formal (finite) linear combinations of elements in $P$;
the set $P$ is a basis of this vector space. We further
consider the $1$-norm on this space, where we define
the $1$-norm of a formal (finite) linear combination
as the sum of absolute values of coefficients in the linear
combination. This norm induces a topology over $V(P)$,
making it a real topological vector space.

Our set of algorithms will be the subset of $V(P)$
consisting of vectors with norm $1$ that have
non-negative coefficients in the linear combination; we denote this
subset by $R$. It is not hard to see that the elements of $R$
are simply all the finite-support probability distributions
over protocols in $P$. We observe that $R$ is a convex set.
This will be the set over which we take the infimum in Sion's
minimax theorem.

Observe that since the input set $\Dom(G)$ of $G$ is finite,
the sets $\Delta_0$ and $\Delta_1$ are both convex, compact
subsets of the real vector space $\bR^{|\Dom(G)|}$, which
has a standard topology. It follows that the
set $\Delta_0\times\Delta_1$ is also convex and compact
(as a subset of the real topological vector space
$\bR^{2|\Dom(G)|}$). This will be the set over which we take the
supremum in Sion's minimax.

Let $A\in R$. This is a finite-support probability
distribution over protocols in $P$; however, it is always
possible to use shared randomness to construct a single
protocol $\Pi_A\in P$ whose behavior exactly matches
that of $A$ (that is, in $\Pi_A$, Alice and Bob will
sample a protocol from $A$ using shared randomness,
and then run that protocol).
Finally, we define
$\alpha\colon R\times(\Delta_0\times\Delta_1)\to[0,\infty)$
by setting
\[\alpha(A,(\mu_0,\mu_1))\coloneqq 
\QIC(\Pi_A,\mu_0)+\QIC(\Pi_A,\mu_1).\]
This will be the function on which we apply Sion's minimax.

It remains to show that $\alpha$ is semicontinuous and
quasisaddle. It is not hard to see that the sum of two
semicontinuous functions (on the same domain) is semicontinuous,
and that the sum of two saddle functions is saddle.
It will therefore be sufficient to show that $\QIC$
is semicontinuous and saddle.

In \cite{Tou15} (Lemma~5), it was shown that $\QIC(\cdot,\mu)$
is linear (and hence convex) for each $\mu$. In \cite{Tou15}
(Lemma 6), it was shown that $\QIC(\Pi,\cdot)$ is concave.
Hence $\QIC$ is saddle, and therefore so is $\alpha$.
In \cite{BGK+18} (Lemma~4.8), it was shown that $\QIC(\Pi,\cdot)$
is continuous.

It remains to show the lower semicontinuity of $\QIC(\cdot,\mu)$.
More explicitly, for each fixed distribution $\mu$, each $A\in R$
and each $\epsilon>0$, there exists $\delta>0$ such that
for all $A'\in R$ with $\|A-A'\|_1<\delta$, we have
$\QIC(\Pi_{A'},\mu)>\QIC(\Pi_A,\mu)-\epsilon$.

We can write $A=(1-p)B+pC$
and $A'=(1-p)B+pC'$ where $B,C,C'\in R$, and $(C,C')$
is a pair of distributions with disjoint support. In other words,
$B$ is the probability distribution consisting of the (normalized)
overlap between $A$ and $A'$, while $C$ and $C'$ are the probability
distributions we get from subtracting out the overlap from $A$
and from $A'$ respectively. If $\|A-A'\|_1<\delta$, we must have
$p<\delta/2$. Now, by the linearity of $\QIC(\cdot,\mu)$,
we have
$\QIC(\Pi_A,\mu)=(1-p)\QIC(\Pi_B,\mu)+p\QIC(\Pi_C,\mu)$ and
$\QIC(\Pi_{A'},\mu)=(1-p)\QIC(\Pi_B,\mu)+p\QIC(\Pi_{C'},\mu)$.
We want to choose $\delta$ so that
$\QIC(\Pi_{A'},\mu)>\QIC(\Pi_A,\mu)-\epsilon$,
or equivalently, so that
$\QIC(\Pi_{C'},\mu) > \QIC(\Pi_C,\mu) - \epsilon/p$.
This rearranges to wanting
$\epsilon/p > \QIC(\Pi_C,\mu)-\QIC(\Pi_{C'},\mu)$;
hence it is sufficient to choose $\delta$ so that
$2\epsilon/\delta > \QIC(\Pi_C,\mu)-\QIC(\Pi_{C'},\mu)$.
It is clear that such $\delta$ can always be chosen,
as $\QIC(\Pi_C,\mu)$ must be finite.

We conclude that $\QIC(\cdot,\mu)$ is lower semicontinuous.
Sion's minimax theorem (\fct{Sion}) then gives
\[\adjustlimits\inf_{A\in R}\sup_{(\mu_0,\mu_1)\in\Delta_0\times\Delta_1}
\QIC(\Pi_A,\mu_0)+\QIC(\Pi_A,\mu_1)
=
\adjustlimits\sup_{(\mu_0,\mu_1)\in\Delta_0\times\Delta_1}\inf_{A\in R}
\QIC(\Pi_A,\mu_0)+\QIC(\Pi_A,\mu_1).\]
Since $R$ contains $P$ as a subset, and since every protocol
in $R$ can be converted into an equivalent protocol in $P$, taking
an infimum over $A\in R$ is the same as taking an infimum over $\Pi\in P$.
It is then clear that the left hand side is at least $\QIC(G)$
(since the latter has only one $\QIC(\Pi,\mu_0)$ or $\QIC(\Pi,\mu_1)$
term instead of both), but no more than twice $\QIC(G)$
(since the maximum of $\QIC(\Pi,\mu_0)$ and $\QIC(\Pi,\mu_1)$
is at least the average of the two).
Hence the desired result follows. The attainment of the maximum
comes from the fact that an upper semicontinuous function
on a nonempty compact set attains is maximum, combined
with the fact that a pointwise infimum of upper semicontinuous
functions is upper semicontinuous.
\end{proof}

\subsection{Product-to-sum reduction for quantum information}
\label{sec:am_gm}

In order to prove \thm{Adv_lift_formal},
we will need a way to bound the
product of quantum information cost on the ``$0$-input'' side
and the quantum information cost on the ``$1$-input'' side.
We start with the following definition.

\begin{definition}
Let $G$ be a communication function. We say a distribution $\mu$
is \emph{nontrivial} for $G$ if for any protocol $\Pi$
computing $G$ (to bounded error against worst-case inputs),
it holds that $\QIC(\Pi,\mu)>1/\poly(r)$, where $r$ is the number
of rounds of $\Pi$.
(In particular, it should not be possible to achieve
$\QIC(\Pi,\mu)=0$ if $\mu$ is nontrivial.)
\end{definition}

Using this definition, we state the following theorem,
which is the main result of this subsection.

\begin{theorem}\label{thm:am_gm}
Let $G$ be a gadget, let $\mu_0$ and $\mu_1$
be nontrivial $0$- and $1$-distributions for $G$,
and let $\Pi$ be a protocol
solving $G$ (to bounded error against worst-case inputs).
Then there is a protocol $\Pi'$ which also solves $G$
(to bounded error against worst-case inputs) which satisfies
\[\QIC(\Pi',\mu_0)+\QIC(\Pi',\mu_1)
=O\left(\sqrt{\QIC(\Pi,\mu_0)\QIC(\Pi,\mu_1)}\cdot \polylog r\right),\]
where $r$ is the number of rounds of $\Pi$ and
where the constant in the big-$O$ is universal.
Moreover, the number of rounds of $\Pi'$ is polynomial in that
of $\Pi$.
\end{theorem}

Before we prove this, we will need a few lemmas.
In the following, we will use $G^{\oplus n}$ to denote the
direct sum of $n$ copies of $G$; that is, if $G\colon\X\times\Y\to\B$,
then $G^{\oplus}\colon(\X^n)\times(\Y^n)\to\B^n$ is the function
that takes in $n$ separate copies to $G$ and outputs $n$
separate outputs from $G$.

\begin{lemma}\label{lem:dir_sum}
Let $(G,\mu_0,\mu_1)$ be any gadget, $0$-distribution,
and $1$-distribution, let $n\in\mathbb{N}^+$, and
let $\Pi$ be a protocol which solves $G^{\oplus n}$
(to bounded error against worst-case inputs). Then
there is a protocol $\Pi'$ which solves $G$
(to bounded error against worst-case inputs) which satisfies
\[\frac{\QIC(\Pi',\mu_0)+\QIC(\Pi',\mu_1)}{2}
\le \frac{1}{n}\cdot \max_{z\in\B^n}
\QIC\left(\Pi,\mu_z\right).\]
\end{lemma}

\begin{proof}
Let for $z\in\B^n$, let $\Pi_{i,z}$ be the protocol which:
takes an input to $G$;
artificially generates $n-1$ inputs from $\mu_{z_j}$ for $j\ne i$
for all the gadgets $G$ except at position $i$; places
the true input at position $i$; runs $\Pi$ on the resulting
input to $G^{\oplus n}$; traces out all the outputs except for
position $i$; and returns the result. Note that $\Pi_{i,z}$
does not depend on the value of $z_i$, but depends on the rest of $z$.
If we use $z^i$ to denote the string $x$ with $i$ flipped, we have
$\Pi_{i,z}=\Pi_{i,z^i}$ for all $x$ and $i$.

\cite{Tou15} (Lemma 4) showed that for all $x\in\B^n$,
\[\sum_{i=1}^n\QIC(\Pi_{i,z},\mu_{z_i})
=\QIC\left(\Pi,\mu_{z}\right).\]
Let $\Pi'\coloneqq\frac{1}{n}\frac{1}{2^n}\sum_{i=1}^n\sum_{x\in\B^n}\Pi_{i,z}$.
Again by \cite{Tou15} (Lemma 5),
\begin{align*}
\frac{\QIC(\Pi',\mu_0)+\QIC(\Pi',\mu_1)}{2}
&=\frac{1}{n2^n}\sum_{i=1}^n\sum_{z\in\B^n}
	\frac{\QIC(\Pi_{i,z},\mu_0)+\QIC(\Pi_{i,z},\mu_1)}{2}\\
&=\frac{1}{n2^n}\sum_{i=1}^n\sum_{z\in\B^n}
	\frac{\QIC(\Pi_{i,z},\mu_{z_i})+\QIC(\Pi_{i,z},\mu_{z^i_i})}{2}\\
&=\frac{1}{n2^n}\sum_{i=1}^n\sum_{z\in\B^n}
	\frac{\QIC(\Pi_{i,z},\mu_{z_i})+\QIC(\Pi_{i,z^i},\mu_{z^i_i})}{2}\\
&=\frac{1}{n2^n}\sum_{i=1}^n\sum_{z\in\B^n}\QIC(\Pi_{i,z},\mu_{z_i})\\
&=\frac{1}{n2^n}\sum_{z\in\B^n}\QIC\left(\Pi,\mu_z\right)\\
&\le \frac{1}{n}\cdot\max_{z\in\B^n}\QIC\left(\Pi,\mu_{z_i}\right).\qedhere
\end{align*}
\end{proof}

\begin{lemma}\label{lem:query_QIC}
Let $G_1,G_2,\dots,G_n$ be any sequence of communication tasks,
and for each $i\in[n]$ let $\Pi_i$ be a protocol which
solves $G_i$ (to bounded error against worst-case inputs).
Let $F$ be a (possibly partial) query function on $n$ bits, and let $Q$
be a $T$-query quantum query algorithm computing $F$ (to bounded error
against worst-case inputs).
Then there is a protocol $\Pi'$ computing $F\circ \{G_i\}_i$ (to bounded
error against worst-case inputs)
such that for any $z\in\Dom(F)$ and any distribution $\mu_z$ supported on
$(G_1\oplus G_2\oplus\dots\oplus G_n)^{-1}(z)$, we have
\[\QIC(\Pi',\mu_z)=\tO\left(T\log\log n\cdot \max_{i\in[n]}\QIC(\Pi_i,\mu^i_z)\right),\]
where $\mu^i_z$ is the marginal of $\mu_z$ on gadget number $i$.
\end{lemma}

\begin{proof}
Let $\hat{\Pi}_i$ be the amplified and purified version of $\Pi_i$,
reducing its worst-case error on $G$ to $\delta/T^{10}\log n$
and using the uncomputing trick to clean up garbage
($\delta$ will be chosen later).
Then the information cost of $\hat{\Pi}_i$
against any fixed distribution increases by a factor
of at most $O(\log T+\log\log n+\log1/\delta)$ compared to $\Pi_i$.
Next, $\Pi'$ be the protocol where Alice runs the query algorithm
for $F$, and whenever she needs to make a query $i$, she sends
$i$ to Bob and they compute gadget number $i$ using $\hat{\Pi}_i$.
Since $F$ succeeds with bounded error on worst-case inputs
and since $\hat{\Pi}_i$ has such a low probability of error,
the protocol $\Pi'$ correctly computes $F\circ G$ on
worst-case inputs.

Fix $z\in\Dom(F)$ and $\mu_z$ supported on $(G^{\oplus n})^{-1}(z)$.
We will expand out $\QIC(\Pi',\mu_z)$.
In round $t\le T$ of the query algorithm, there
are two types of messages between Alice and Bob: one
message from Alice to Bob containing a copy $E_t$ of the query
register $D_t$ for step $t\le T$, which Alice knows from her simulation
of the algorithm $Q$ for $F$; and
all the messages between Alice and Bob implementing $\hat{\Pi}_i$.
Denote those messages by $C_{t,j}$. Note that
$E_t$ also gets passed back from Bob to Alice at the end of each
round for cleanup purposes.

We name the rest of the registers.
Let the input registers be $X$ and $Y$, and let Alice hold
register $D_t$ specifying the position
to query at round $t$, a work register $\tilde{A}_t$
related to the implementation of the algorithm
$Q$ for $f$ (which stays untouched for all $j$),
and register $A_{t,j}$ related to the implementation
of the $\hat{\Pi}_i$ protocols
for round $t$. Bob holds query register $E_t$
(passed from Alice, untouched for all $j$)
as well as work register $B_{t,j}$
for the implementation of the $\hat{\Pi}_i$ protocols.
Let $R$ be the purification register. Then using $r$
to denote the index of the last round of the $\hat{\Pi}_i$, we have
\begin{align*}
\QIC(\Pi',\mu_z)&=\sum_{t=1}^T I(\tX\tY:E_t|YB_{t,0})_{\Psi^t_z}
+\sum_{t=1}^T I(\tX\tY:E_t|XA_{t,r}\tilde{A}_tD_t)_{\Psi^t_z}\\
&+\sum_{t=1}^T\sum_{j\text{ odd}} I(\tX\tY:C_{t,j}|YB_{t,j}E_t)_{\Psi^{t,j}_z}
+\sum_{t=1}^T\sum_{j\text{ even}} I(\tX\tY:C_{t,j}|XA_{t,j}D_t\tilde{A}_t)_{\Psi^{t,j}_z}.
\end{align*}

For the terms $I(\tX\tY:E_t|YB_{t,0})_{\Psi^t_z}$ and $I(\tX\tY:E_t|XA_{t,r}\tilde{A}_tD_t)_{\Psi^t_z}$, we note that $B_{t,0}$ and $A_{t,r}$ are the start state on Bob's side and the end state on Alice's side for $\hat{\Pi}$, and can be assumed to be independent of all other registers. Hence we shall ignore the registers $B_{t,0}$ and $A_{t,r}$ in the conditioning systems. Let $\ket{\Phi^t}$ that is obtained by replacing the $\tilde{A_t}D_tE_t$ registers of $\ket{\Psi^t_z}$ with the state of the query algorithm for $f$ after $t$ queries (with the query register $D_t$ duplicated). $\ket{\Phi^t}_{\tilde{A}_tD_tE_t|zxy}$ depends on $x$ and $y$ only through $z$, which is fixed. Hence $I(\tX\tY:E_t|Y)_{\Phi^t_z}$ and $I(\tX\tY:E_t|X\tilde{A}_t)_{\Phi^t_z}$ are both 0. Clearly, $\Phi^t_z$ is the state the protocol would have been in if $\hat{\Pi}_i$ were run with 0 error. Since the protocol runs of $\hat{\Pi}_i$
make very small error instead, we have instead $\norm{\ket{\Psi^t}_z - \ket{\Phi^t}_z}_1 \leq \epsilon$,
where $\epsilon=O(\delta/\poly(T)\log n)$. This implies
\begin{align*}
I(\tX\tY:E_t|Y)_{\Psi^t_z} & = H(E_t|Y)_{\Psi^t_z} - H(E_t|\tX\tY Y)_{\Psi^t_z} \\
 & \leq H(E_t|Y)_{\Phi^t_z} - H(E_t|\tX\tY Y)_{\Phi^t_z} + 8\epsilon\log|E_t| + 4h(\epsilon) \\
 & = 8\epsilon + 4h(\epsilon).
\end{align*}
The total sum of
$I(\tX\tY:E_t|YB_{t,0})_{\Psi^t_z}$ over all $t$ is therefore at most $\delta/2$, and the same applies to $I(\tX\tY:E_t|X\tilde{A}_t)_{\Psi^t_z}$.

We then have
\begin{align*}
\QIC(\Pi',\mu_z)&\le \delta+\sum_{t=1}^T
\left(\sum_{j\text{ odd}} I(\tX\tY:C_{t,j}|YB_{t,j}E_t)_{\Psi^t_z}
    +\sum_{j\text{ even}} I(\tX\tY:C_{t,j}|XA_{t,j}D_t\tilde{A}_t)_{\Psi^t_z}\right)\\
&=\delta+ \sum_{t=1}^T\sum_{i=1}^n\Pr[D_t=i]\left(
    \sum_{j\text{ odd}} I(\tX\tY:C_{t,j}|YB_{t,j})_{\Psi^t_{z,D_t=i}} \right. \\
&   \quad \left. +\sum_{j\text{ even}} I(\tX\tY:C_{t,j}|XA_{t,j}\tilde{A}_t)_{\Psi^t_{z,D_t=i}}\right)\\
&=\delta+\sum_{t=1}^T\sum_{i=1}^n \Pr[D_t=i]\QIC(\hat{\Pi}_i,\mu^i_z)\\
&\le\delta+ T\max_i\QIC(\hat{\Pi}_i,\mu^i_z)\\
&\le\delta+ O(T(\log T+\log\log n+\log 1/\delta)\max_i\QIC(\Pi_i,\mu^i_z)).
\end{align*}
Setting $\delta= T\max_i\QIC(\Pi_i,\mu^i_z)$, but ensuring
$\epsilon\le 1/3$ (since we can't amplify a negative amount), we get
\[\QIC(\Pi',\mu_z)=O\left(T\max_i\QIC(\Pi_i,\mu^i_z)
\log\left(2+\frac{T^{10}\log n}{\max_i\QIC(\Pi_i,\mu^i_z)}\right)\right).
\qedhere\]
\end{proof}

\begin{lemma}\label{lem:all_to_OR}
Let $G$ be a gadget, let $\mu_0$ and $\mu_1$ be a $0$-distribution
and a $1$-distribution for $G$, let $n\in\mathbb{N}^+$, and
let $\Pi$ be a protocol computing $\OR_n\circ G$ (to bounded
error against worst-case inputs). Then there is a protocol $\Pi'$
computing $G^{\oplus n}$
(to bounded error against worst-case inputs) such that
\[\max_{z\in\B^n}\QIC\left(\Pi',\mu_z\right)
=\tO\left(\sqrt{n}\cdot \max_{z\in\B^n}\QIC\left(\Pi,\mu_z\right)\right).\]
\end{lemma}

\begin{proof}
Consider the following task: the goal is to output
a hidden string $z\in\B^n$, and the allowed queries are subset-OR
queries, meaning that for each subset $S\subseteq[n]$ there is
a query which returns $\OR(z_S)$ (which equals $1$ if $z_i=1$ for
some $i\in S$, and returns $0$ otherwise). We can model
this task as a query function $F$ on a promise set $P\subseteq\B^{2^n}$.
Each string in $P$ is a long encoding $u(z)\in\B^{2^n}$
of some string $z\in\B^n$,
with the long encoding $u(z)$ being a string with $(u(z))_S=\OR(z_S)$
for all $S$. In other words, $u$ is a function $u\colon\B^n\to\B^{2^n}$.
The function $F$ is defined by $F(u(z))=z$ for all
$z\in\B^n$, where $\Dom(F)=\{u(z):z\in\B^n\}$.
It is not hard to verify that this function is well-defined.

The function $f$ is sometimes called the combinatorial group testing
problem. We have $\D(F)\le n$, since we can query $u(z)_{\{i\}}$
for all $i\in[n]$ to get the bits $z_i$ one by one and then output
all of $z$. (Note that the input size to $F$ is of length $N=2^n$,
so $n$ does not represent the input size here.)
Belovs \cite{Bel15} showed that $\Q(F)=O(\sqrt{n})$.
This result will play a key role in our analysis here,
which is motivated by \cite{BBGK18}
(where this algorithm of Belovs was similarly used to reduce
direct-sum computations to OR computations).

Now, observe that $F\circ u$ is the identity function on $n$ bit strings.
The protocol $\Pi'$ for $G^{\oplus}$ will be a protocol
for $F\circ u\circ G$. We use \lem{query_QIC}
on the query function $F$ and the communication tasks
$u(G^{\oplus n})_1,u(G^{\oplus n})_2,\dots,u(G^{\oplus n})_{2^n}$.
The query algorithm for $F$ makes $T=O(\sqrt{n})$ queries.
Each of the communication tasks is of the following form:
take as input $n$ copies to $G$, and output the OR
of a fixed subset $S$ of the copies of $G$. To solve this task,
which we denote $F_S$, we describe a protocol $\Pi_S$.
In this protocol, Alice and Bob will use their shared randomness
to replace the inputs in positions $i\notin S$ by independent
samples from $\mu_0$. They will then run $\Pi$ to compute the OR
of the $n$ copies of $G$.

The correctness of $\Pi'$ is clear, so we analyze its information cost.
Fix $z\in\B^n$, and denote by $z_S$ the string satisfying
$(z_S)_i=z_i$ if $i\in S$ and $(z_S)_i=0$ if $i\in S$.
In order to upper bound $\QIC(\Pi',\mu_z)$ using \lem{query_QIC},
we let $\mu'_z$ be the distribution on strings of length
$\B^{n 2^n}$ that we get by sampling a string from $\mu_z$
and making $2^n$ copies of it.
We observe that that the behavior of $\Pi'$ when acting on $\mu_z$
is exactly the composed behavior of the query algorithm for $F$
composed with the protocols $\Pi_S$ acting on the distribution
$\mu'_z$; \lem{query_QIC} therefore gives us
\[\QIC(\Pi',\mu_z)=O\left(\sqrt{n}\cdot
\max_S\QIC(\Pi_S,\mu_z)
\log\left(2+\frac{n^5\log N}{\max_S\QIC(\Pi_S,\mu_z)}\right)\right)\]
(where we used the more precise bound given in the proof
of \lem{query_QIC}).
Recall that $\Pi_S$ replaces the samples of $\mu_z$
that correspond to bits $i\notin S$ with freshly-generated
samples from $\mu_0$, and then runs $\Pi$; hence
$\QIC(\Pi_S,\mu_z)=\QIC(\Pi_S,\mu_{z_S})\le \QIC(\Pi,\mu_{z_S})$.
The maximum over sets $S$ of $\QIC(\Pi,\mu_{z_S})$
is clearly at most the maximum over $w\in\B^n$ of $\QIC(\Pi,\mu_w)$.
Using $\log N=n$, we can therefore write
\[\QIC(\Pi',\mu_z)=O\left(\sqrt{n}\cdot
\max_w\QIC(\Pi,\mu_w)
\log\left(2+\frac{n^6}{\max_w\QIC(\Pi,\mu_w)}\right)\right).\qedhere\]
\end{proof}

\begin{lemma}\label{lem:OR_protocol}
Let $G$ be a gadget, let $\mu_0$ and $\mu_1$ be a $0$-distribution
and a $1$-distribution for $G$, and let $\Pi$ be a protocol
computing $G$ (to bounded error against worst-case inputs).
Then for any $n\in\mathbb{N}^+$, there is a protocol $\Pi'$
computing $\OR_n\circ G$ (to bounded error against worst-case inputs)
such that
\[\max_{z\in\B^n}\QIC\left(\Pi',\mu_z\right)
=O(n\QIC(\Pi,\mu_0)+\log n\cdot \QIC(\Pi,\mu_1)).\]
\end{lemma}

\begin{proof}
In order to compute $\OR_n\circ G$, the
protocol $\Pi'$ will simply compute each copy of $G$ one at a time,
stopping as soon as a $1$ has been found. The idea is that this
will ensure the number of computations of $0$-inputs to $G$
is at most $O(n)$ while the number of computations of
$1$-inputs to $G$ is $\tO(1)$.

To be more formal, we consider a cleaned up version
$\hat{\Pi}$ of $\Pi$, which will have error $O(1/n)$
and which cleans up all the garbage and resets Alice and Bob's
states to their initial states after the computation is complete.
The protocol $\Pi'$ will run $\hat{\Pi}$ on each input to $G$,
in sequence, stopping when an output $1$ has been found.
To implement this, we will name the registers:
suppose the protocol $\hat{\Pi}$ uses registers $A$ and $O_A$
on Alice's side and registers $B$ and $O_B$ on Bob's side,
where $O_A$ and $O_B$ store the final output of $\hat{\Pi}$.
At the beginning of $\hat{\Pi}$, the registers are expected
to be $\ket{0}_A\ket{0}_B\ket{0}_{O_A}\ket{0}_{O_B}$.
The guarantee of $\hat{\Pi}$ is that at the end of the algorithm,
the registers will be in the state
$\ket{0}_A\ket{0}_B\ket{b}_{O_A}\ket{b}_{O_B}$,
where $b$ is close to the output of $G$ on that input.
We now implement $\Pi'$ by adding an additional register
on each side, denoted $S_A$ and $S_B$, which stores the strings
of outputs of all the runs of $\hat{\Pi}$. These
registers are each initialized to $0^n$. At the
end of run $i$ of $\hat{\Pi}$ (which computes gadget $i$),
Alice and Bob will each swap the register $O_A$ with
the $i$-th bit of $S_A$; this resets the registers
used by $\hat{\Pi}$ to be all zero, and it stores the output
of the $i$-th run of $\hat{\Pi}$ so that $\Pi'$ has access
to it. It also preserves the property that $S_A=S_B$ throughout
the algorithm.

The final detail is that in $\Pi'$, Alice and Bob only run
$\hat{\Pi}$ on gadget $i$ if they see that all the previous runs
resulted in output $0$; that is, they control the implementation
of $\hat{\Pi}$ on the registers $S_A$ and $S_B$ being equal to $0^n$.
This will ensure that once a $1$ is found, no further information
will be exchanged between Alice and Bob. The final output
of $\Pi'$ will be $0$ if $S_A$ and $S_B$ are $0^n$, and it will be
$1$ otherwise.

The correctness of $\Pi'$ (to worst-case bounded error) is clear,
so we analyze its information cost against $\mu_z$ for a
fixed string $z\in\B^n$. The information cost $\QIC(\Pi',\mu_z)$
is a sum of information exchanged over all rounds;
let $\QIC_i(\Pi',\mu_z)$ denote the sum of information
exchanged only in the rounds corresponding to the computation
of the $i$-th copy of $G$, so that
$\QIC(\Pi',\mu_z)=\sum_{i=1}^n\QIC_i(\Pi',\mu_z)$.

Let $T$ be the number of rounds used by $\hat{\Pi}$.
Let $S_{A,i}$ and $S_{B,i}$ be the registers $S_A$ and $S_B$
during the computation of the $i$-th copy of $G$. 
Use $X$ and $Y$ to denote Alice and Bob's inputs respectively,
with $X_i$ and $Y_i$ being the inputs to copy $i$ of $G$
and with $\tilde{X}$ and $\tilde{Y}$ denoting their purifications,
and let $C$ be the register passed back and forth between
Alice and Bob in $\hat{\Pi}$. Then
\[\QIC_i(\Pi',\mu_z)=
\sum_{t\le T\text{ odd}} I(\tX\tY:C_t|YB_t S_{B,i})
+\sum_{t\le T\text{ even}}I(\tX\tY:C_t|XA_t S_{A,i}).\]
We note that the register $S_{B,i}$ in the odd terms is classical,
as is the register $S_{A,i}$. Hence the conditional
mutual information conditioned on $S_{B,i}$ is the expectation
of the conditional mutual information conditioned on the events
$S_{B,i}=w$ for each string $w\in\B^n$ (see, for example,
\cite{BGK+18}, end of Section 3.1). In other words,
\[I(\tX\tY:C_t|YB_t S_{B,i})=
\Pr[S_{B,i}=0^n] I(\tX\tY:C_t|YB_t)_{S_{B,i}=0^n}
+\Pr[S_{B,i}\ne 0^n] I(\tX\tY:C_t|YB_t)_{S_{B,i}\ne 0^n}.\]
Note that by the construction of $\Pi'$, in the second term
we actually have $I(\tX\tY:C_t|YB_t)_{S_{B,i}=\ne 0^n}=0$,
since the registers $C_t$ are all $0$ as Alice and Bob do not
run $\hat{\Pi}$ at all when $S_{B,i}\ne 0^n$.
The term $I(\tX\tY:C_t|YB_t)_{S_{B,i}=0^n}$
is just $I(\tX_i\tY_i:C_t|Y_iB_t)$, since
the run of $\hat{\Pi}$ ignores everything outside of the input
to the $i$-th copy of $G$. Hence we have
\begin{align*}
\QIC_i(\Pi',\mu_z)&=
\Pr[S_{B,i}=0^n]\sum_{t\le T\text{ odd}}I(\tX_i\tY_i:C_t|Y_iB_t)
+\Pr[S_{A,i}=0^n]
\sum_{t\le T\text{ even}}I(\tX_i\tY_i:C_t|X_iA_t)\\
&=\Pr[S_{A,i}=0]\QIC(\hat{\Pi},\mu_{z_i}).
\end{align*}
From this, it follows that
\[\QIC(\Pi',\mu_z)=\sum_{i=1}^n \Pr[S_{A,i}=0^n]\QIC(\hat{\Pi},\mu_{z_i}).\]

To upper bound this, we note that the total sum of
all the terms $\Pr[S_{A,i}=0^n]\QIC(\hat{\Pi},\mu_{z_i})$ for
$i$ such that $z_i=0$ is at most $n\QIC(\hat{\Pi},\mu_0)$,
where we've upper bounded $\Pr[S_{A,i}=0^n]\le 1$.
For $i$ such that $z_i=1$, we split into two cases:
in the case where $i$ is the first index such that $z_i=1$,
we upper bound
$\Pr[S_{A,i}=0^n]\QIC(\hat{\Pi},\mu_{z_i})\le \QIC(\hat{\Pi},\mu_1)$.
In contrast, for all $i$ such that $z_i=1$ and for which
there was a previous index $j<i$ with $z_j=1$, we note that
the $1/n$ error guarantee of $\hat{\Pi}$ ensures that
$\Pr[S_{A,i}=0^n]\le 1/n$; hence these terms are individually
at most $(1/n)\QIC(\hat{\Pi},\mu_1)$, and the sum of all of them
is at most $\QIC(\hat{\Pi},\mu_1)$. We conclude that
\[\QIC(\Pi',\mu_z)\le n\QIC(\hat{\Pi},\mu_0)+2\QIC(\hat{\Pi},\mu_1).\]
Finally, we note that $\hat{\Pi}$ simply repeats $\Pi$
$O(\log n)$ times and takes a majority votes in order to amplify
(and then runs this in reverse to clean up garbage). Hence
we have
\begin{align*}
\QIC(\hat{\Pi},\mu_0)&=O(\log n\cdot \QIC(\Pi,\mu_0)),\\
\QIC(\hat{\Pi},\mu_1)&=O(\log n\cdot \QIC(\Pi,\mu_1)).
\end{align*}
This gives the upper bound on $\QIC(\Pi',\mu_z)$ of
$O(n\log n\cdot \QIC(\Pi,\mu_0)+\log n\cdot \QIC(\Pi,\mu_1))$.

Finally, we sketch how to shave the log factor from the $\QIC(\Pi,\mu_0)$
term. To do so, we avoid amplifying $\hat{\Pi}$. Instead,
we simply run $\hat{\Pi}$ on each input. If the output is $1$,
we run $\hat{\Pi}$ again on the same copy of $G$. We do so until
the number of $0$ outputs outnumbers the number of $1$ outputs.
If $O(\log n)$ repetitions happened and the number of $1$-outputs
is still larger than the number of $0$ inputs, we finally
``believe'' that this gadget evaluates to $1$ and halt.
Otherwise, if the $0$s outnumber the $1$s before that point,
then we assume the gadget evaluated to $0$ and move on to the
next one.

By analyzing this as the ``monkey on a cliff'' problem,
it is not hard to see that a $1$ gadget is correctly
labelled as such with constant probability.
The total number of runs of $\hat{\Pi}$ on $0$-inputs
will, on expectation, be at most $O(n)$, while the total
number of runs of $\hat{\Pi}$ on $1$-inputs will be
at most $O(\log n)$ on expectation; since we avoided
the $O(\log n)$ loss from amplification, this protocol
is more efficient, and we shave a log factor from the
$\QIC(\Pi,\mu_0)$ dependence.\footnote{We thank
Thomas Watson and Mika G{\"o}{\"o}s for pointing out
this ``monkey on a cliff'' strategy for computing OR on a noisy oracle.}
\end{proof}

We are now ready to prove \thm{am_gm}.

\begin{proof} (of \thm{am_gm}.)
Using \lem{OR_protocol}, we get a protocol $\Pi_2$
computing $\OR_n\circ G$ such that for any $z\in\B^n$,
$\QIC(\Pi_2,\mu_z)=O(n\QIC(\Pi,\mu_0)+\log n\cdot\QIC(\Pi,\mu_1))$.
Using \lem{all_to_OR}, we get a protocol $\Pi_3$
computing $G^{\oplus n}$ such that for any $z\in\B^n$,
\begin{multline*}
\QIC(\Pi_3,\mu_z)=\\
O\left((n^{3/2}\cdot\QIC(\Pi,\mu_0)+
\sqrt{n}\log n\cdot\QIC(\Pi,\mu_1))\log\left(2+\frac{n^6}
{n\QIC(\Pi,\mu_0)+\log n\cdot\QIC(\Pi,\mu_1)}\right)\right).
\end{multline*}
Finally, using \lem{dir_sum}, we get a protocol $\Pi_4$
computing $G$ such that
\begin{multline*}
\QIC(\Pi_4,\mu_0)+\QIC(\Pi_4,\mu_1)=\\
O\left(\left(\sqrt{n}\cdot\QIC(\Pi,\mu_0)+
\frac{\log n}{\sqrt{n}}\cdot\QIC(\Pi,\mu_1)\right)
\log\left(2+\frac{n^6}
{n\QIC(\Pi,\mu_0)+\log n\cdot\QIC(\Pi,\mu_1)}\right)\right).
\end{multline*}
Moreover, by negating the output of $G$, such a protocol
$\Pi_4$ also exists with the $\mu_0$ and $\mu_1$ reversed.

Now, assume without loss of generality that
$\QIC(\Pi,\mu_0)\le\QIC(\Pi,\mu_1)$.
Let $\ell$ be the ratio $\QIC(\Pi,\mu_1)/\QIC(\Pi,\mu_0)\ge 1$
(here we use the assumption that $\QIC(\Pi,\mu_0)>0$
and that $\QIC(\Pi,\mu_1)>0$).
Let $n\in\mathbb{N}^+$ be $\lceil2\ell\log2\ell\rceil$. Note that $n$
is at most $3\ell\log 2\ell$, so $n=\Theta(\ell\log 2\ell)$
and $\log n=\Theta(\log 2\ell)$. Using this value of $n$,
we get $\Pi'$ such that
\begin{multline*}
\QIC(\Pi',\mu_0)+\QIC(\Pi',\mu_1)=\\
O\left(\sqrt{\QIC(\Pi,\mu_0)\QIC(\Pi,\mu_1)}\log^{1/2}
\frac{\QIC(\Pi,\mu_0)+\QIC(\Pi,\mu_1)}
    {\sqrt{\QIC(\Pi,\mu_0),\QIC(\Pi,\mu_1)}}
    \cdot\log(2+\alpha)\right),
\end{multline*}
where 
\[\alpha=\frac{(\QIC(\Pi,\mu_0)+\QIC(\Pi,\mu_1))^{11}}
    {\QIC(\Pi,\mu_0)^6\QIC(\Pi,\mu_1)^6}.\]

If $\mu_0$ and $\mu_1$ are nontrivial, so that we have (say)
$\QIC(\Pi,\mu_0)>1/r^{10}$ and $\QIC(\Pi,\mu_1)>1/r^{10}$,
this can be simplified to
\[\QIC(\Pi',\mu_0)+\QIC(\Pi,\mu_1)
=O(\sqrt{\QIC(\Pi,\mu_0)\QIC(\Pi,\mu_1)}\log^{3/2} r).\]
Finally, since $n$ is at most polynomial in $r$, it is not
hard to check that each of these reductions increases the number
of rounds by only a polynomial factor in $r$, so the final
protocol $\Pi'$ has number of rounds $\poly(r)$.
\end{proof}

\subsection{Proving the lifting theorem}

\Advliftformal*

\begin{proof}
Let $\mu_0'$ and $\mu_1'$ be the distributions
for $G$ provided by \thm{QICZ_minimax}. Let $\mu_0$
be the equal mixture of $\mu_0'$ and the uniform distribution
over $0$-inputs to the $\AND_2$ gadget inside of $G$,
and let $\mu_1$ be the equal mixture of $\mu_1'$ and the
uniform distribution over the $1$-inputs to the $\OR_2$ gadget
inside of $G$. We note that for any protocol $\Pi$,
$\QIC(\Pi,\mu_0)\ge \QIC(\Pi,\mu_0')/2$ and
$\QIC(\Pi,\mu_1)\ge \QIC(\Pi,\mu_1')/2$. By \cite{BGK+18}, if $\Pi$ has $r$ rounds,
$\QIC(\Pi,\mu_0),\QIC(\Pi,\mu_1)=\Omega(1/r)$. So $\mu_0$ and $\mu_1$ are nontrivial for $G$.

Let $\Pi$ be a protocol computing $f\circ G$ to error $\epsilon$,
and let $r$ be the number of rounds used by $\Pi$,
and let $T$ be the communication cost of $\Pi$.
For $z\in\B^n$, we define 
\[ q'(z,i)\coloneqq \sum_{t \text{ odd}}I(\tX_i\tY_i:C_t|\tX_{<i}\tY_{<i}B'_t)_{\mu_z} + \sum_{t\text{ even}}I(\tX_i\tY_i:C_t|\tX_{>i}\tY_{>i}A'_t)_{\mu_z}\]
where $C_t$ is the message in the $t$-th round of $\Pi$ and $A'_t, B'_t$ are Alice and Bob's memory registers (which don't necessarily have safe copies of their inputs). By the chain rule of mutual information, we have
\[\sum_{i=1}^n q'(z,i) = \QIC(\Pi, \mu_z) \leq T\]
for all $z\in\B^n$. A feasible weight scheme $q(z,i)$ for $\Adv_1(f)$ will be defined by normalizing $q'(z,i)$.

Let $z,w\in\B^n$ be such that $f(z)$ and $f(w)$ are disjoint,
and such that their Hamming distance is $1$. Let
$i\in[n]$ be the bit on which they disagree, so that
$z^i=w$ (where $z^i$ denotes the string $z$ with bit $i$ flipped).
Suppose without loss of generality that $z_i=1$ and $w_i=0$.
In order to lower bound $q'(z,i)\cdot q'(w,i)$, we will use the protocol $\Pi$ for $f\circ G$ to construct a protocol $\Pi'$ for $G$.

The protocol $\Pi'$ is given by \cite{Tou15} (Lemma 4). Alice and Bob start with the shared entangled state of $\Pi$, as well the $\tX_{-i}X_{-i}\tY_{-i}Y_{-i}$ registers of their inputs and purification according to $\mu_{z_{-i}}$ (=$\mu_{w_{-i}}$) in $\Pi$, with Alice holding $A_0\tX_{<i}\tY_{<i}X_{-i}$ and Bob holding $B_0\tX_{>i}\tY_{>i}Y_{-i}$ (here $X_{-i}$ denotes $X_1\ldots X_{i-1}X_{i+1}\ldots X_n$ and the same is true for other variables). They will embed their inputs for $\Pi'$, which we call $X',Y'$ (with purifications $\tX'\tY'$), into the $i$-th input register for $\Pi$ (with $\tX', \tY'$ being embedded as $\tX_i, \tY_i$), and use their shared entanglement for the rest of the input registers, to run $\Pi$. After running $\Pi$, they will output $1$ if
$\Pi$ outputs a symbol in $f(z)$ (outputting $0$ otherwise).
Note that since $\Pi$ outputs a symbol in $f(z)$
with probability at least $1-\epsilon$ when given an input
from $(G^{\oplus n})^{-1}(z)$ and with probability at most $\epsilon$
when given an input from $(G^{\oplus n})^{-1}(w)$
(since $f(w)\cap f(z)=\varnothing$), it follows that
$\Pi'$ succeeds to error $\epsilon$ on worst-case inputs
to $G$.

We now analyze the information cost of $\Pi'$. Against the distribution $\mu_0$,
\begin{align*}
\QIC(\Pi', \mu_0) & = \sum_{t \text{ odd}}I(\tX_i\tY_i:C_t|\tX_{<i}\tY_{<i}B'_t)_{\mu_{w_{-i}}\otimes \mu_{w_i}} + \sum_{t\text{ even}}I(\tX_i\tY_i:C_t|\tX_{>i}\tY_{>i}A'_t)_{\mu_{w_{-i}}\otimes\mu_{w_i}} \\
 & = \sum_{t \text{ odd}}I(\tX_i\tY_i:C_t|\tX_{<i}\tY_{<i}B'_t)_{\mu_w} + \sum_{t\text{ even}}I(\tX_i\tY_i:C_t|\tX_{>i}\tY_{>i}A'_t)_{\mu_w} = q'(w,i).
\end{align*}
Similarly, $\QIC(\Pi',\mu_1)=q(z,i)$,
so we have
\[\sqrt{q'(z,i)q'(w,i)}
=\sqrt{\QIC(\Pi',\mu_0)\QIC(\Pi',\mu_1)}.\]

By \thm{am_gm}, there is a protocol $\Pi''$ such that
\[\sqrt{q'(z,i)q'(w,i)}= \Omega\left(\frac{\QIC(\Pi'',\mu_0)+\QIC(\Pi'',\mu_1)}
{\polylog r}\right).\]
By the choice of $\mu_0$ and $\mu_1$, we therefore have
\[\sqrt{q'(z,i)q'(w,i)}=\Omega(\QICZ(G)/\polylog r),\]
and hence by taking $q(z,i)=O(\polylog r/\QICZ(G))\cdot q'(z,i)$,
we get $q(z,i)q(w,i)\ge 1$. If we start with a protocol $\Pi$
with number of rounds $r$ at most $\QCC_{\epsilon}(f\circ G)$,
we conclude
\[\QCC_\epsilon(f\circ G)=\tOmega(\Adv_1(f)\QICZ(G)),\]
as desired.
\end{proof}

\section{New query relations}\label{sec:query}

In this section, we prove our new relationships in
query complexity. We start by showing that $\cfbs(f)$
is equivalent to $\CAdv(f)$ for partial functions.
To do so, we will first need the well-known dual
form for the fractional block sensitivity
at a specific input, $\fbs(x,f)$. This dual form can
be derived by writing the weight scheme defining
$\fbs(x,f)$ as a linear program, and taking the dual;
this gives a minimization program in which
$\fbs(x,f)$ is the minimum, over weight schemes
$q(i)\ge0$ assigned to each $i\in[n]$ that satisfy
$\sum_{i\in B} q(i)\ge 1$ for each sensitive block
$B\subseteq[n]$ of $x$ (with respect to $f$),
of the sum $\sum_{i\in[n]} q(i)$.
See any of \cite{Aar08,Tal13,KT16} for a formal proof.

\begin{lemma}
$\cfbs(f)\le 2\CAdv(f)$.
\end{lemma}

\begin{proof}
Let $q(x,i)$ be a feasible weight scheme for $\CAdv(f)$
with objective value equal to $\CAdv(f)$.
We construct a completion $f'$ of $f$ as follows.
For each $z\notin\Dom(f)$, let $z'\in\Dom(f)$ be the
input in the domain of $f$ which minimizes
$\sum_{i:z'_i\ne z_i} q(z',i)$. Set $f'(z)=f(z')$.
Now let $x$ be any input in $\Dom(f)$;
we wish to upper bound $\fbs(x,f')$.

To this end, we pick weights $q(i)=2q(x,i)$, and claim
that they are a feasible solution to the fractional
block sensitivity for $f'$ at $x$. Let $B$ be any
sensitive block for $x$ with respect to $f'$.
Then $x^B$ is some input $z$ which disagrees with $x$
exactly on the bits in $B$, and which
satisfies $f'(z)\ne f(x)$. Let $z'$ be the input in $\Dom(f)$
which minimizes $\sum_{i:z'_i\ne z_i} q(z',i)$, so that
$f'(z)=f(z')$. Then $f(z')\ne f(x)$, and in fact $z'$ must be
closer to $z$ than to $x$; hence
\begin{align*}
\sum_{i\in B} q(i) & = \sum_{i:x_i\ne z_i} 2q(x,i) \\
 & \ge \sum_{i:x_i\ne z_i} q(x,i) + \sum_{i:z'_i\ne z_i} q(z',i) \\
 & \ge \sum_{i:x_i\ge z_i} \min\{q(x,i),q(z',i)\}
+\sum_{i:z'_i\ne z_i}\min\{q(x,i),q(z',i)\} \\
 & \ge \sum_{i:x_i\ne z'_i}\min\{q(x,i),q(z',i)\}\ge 1.
\end{align*}
We conclude that $q(i)$ is feasible. Its objective value
is $\sum_{i\in[n]} q(i)=\sum_{i\in[n]} 2q(x,i)\le 2\CAdv(f)$,
and hence $\cfbs(f)\le 2\CAdv(f)$, as desired.
\end{proof}

\begin{lemma}[Kri{\v s}j{\=a}nis Pr{\=u}sis, personal communication]\label{lem:CAdv_cfbs}
$\CAdv(f)\le \cfbs(f)$.
\end{lemma}

\begin{proof}
Let $f'$ be a completion of $f$ for which
$\fbs(x,f')\le \cfbs(f)$ for all $x\in\Dom(f)$.
For each $x\in\Dom(f)$, let $q_x(i)$ be a feasible
weight scheme for the minimization problem of $\fbs(x,f')$
which satisfies $\sum_{i\in[n]} q_x(i)\le \fbs(x,f')\le \cfbs(f)$
and for each sensitive block $B$ of $f'$,
$\sum_{i\in B} q_x(i)\ge 1$.

We construct a weight scheme for $\CAdv(f)$ by setting
$q(x,i)=2q_x(i)$ for all $x\in\Dom(f)$.
We claim this weight scheme is feasible. To see this,
let $x,y\in\Dom(f)$ be such that $f(x)\ne f(y)$.
Define the input $z\in\B^n$ such that $z_i=x_i$
if $x_i=y_i$, and otherwise, $z_i=x_i$ if $q(x,i)\ge q(y,i)$
and $z_i=y_i$ if $q(y,i)> q(x,i)$.
Suppose that $f'(z)\ne f(x)$. Then
\begin{align*}
\sum_{i:x_i\ne y_i}\min\{q(x,i),q(y,i)\}
& =\sum_{i:x_i\ne z_i}\min\{q(x,i),q(y,i)\}
+\sum_{i:y_i\ne z_i}\min\{q(x,i),q(y,i)\} \\
& \ge\sum_{i:x_i\ne z_i}q(x,i)
+\sum_{i:y_i\ne z_i}q(y,i)\ge 1. \qedhere
\end{align*}
\end{proof}

\begin{lemma}\label{lem:adeg_fbs}
For any (possibly partial) Boolean function $f$, we have
\[\adeg_\epsilon(f)\ge
\frac{\sqrt{2}}{\pi} \sqrt{(1-2\epsilon)\fbs(f)}.\]
\end{lemma}

\begin{proof}
Let $x\in\Dom(f)$ be such that $\fbs(x,f)=\fbs(f)$.
By negating the input bits of $f$ if necessary,
we may assume that $x=0^n$ (note that negating
input bits does not affect $\fbs(f)$ or $\adeg(f)$).
By negating the output of $f$ if necessary, we can
further assume that $f(0^n)=0$. Let $p$ be a polynomial
of degree $\adeg_\epsilon(f)$ which approximates
$f$ to error $\epsilon$.

Let $\PrOR_k$ be the promise problem on $k$
bits whose domain contains all the strings
of Hamming weights $0$ or $1$, and which
outputs $0$ given $0^k$ and outputs $1$ given an input
of Hamming weight $1$.

We give an exact polynomial representation of this function.
To do so, let $T_d$ be the Chebyshev polynomial of degree $d$;
this is the single-variate real polynomial satisfying
$T_d(\cos\theta)=\cos(d\theta)$.
This polynomial is bounded in $[-1,1]$ on the interval
$[-1,1]$. Moreover, it satisfies $T_d(1)=1$
and $T_d(\cos(\pi/d))=-1$. Hence the polynomial
$r(t)=(1-T_d(1-(1-\cos(\pi/d))t))/2$ evaluates to $0$
at $t=0$ and to $1$ at $t=1$. Moreover, since
this $T_d$ is bounded in $[-1,1]$ on the interval $[-1,1]$,
we conclude that $r(t)$ is bounded in $[0,1]$
on the interval $[0,2/(1-\cos(\pi/d))]$.
Since $\cos(z)\ge 1-z^2/2$, we have
$2/(1-\cos(\pi/d))\ge 4d^2/\pi^2$.
Hence $r(t)$ is bounded in $[0,1]$ on the interval
$[0,4d^2/\pi^2]$. If we pick $d$ such that
$4d^2/\pi^2\ge k$, that is, $d$ at least
$\lceil \pi \sqrt{k}/2\rceil$, then we would know
that $r(t)$ is bounded on $[0,k]$.
In that case, the $k$-variate
polynomial $q(x)=r(x_1+x_2+\dots+x_k)$
would exactly compute $\PrOR_k$,
and it would have degree at most 
$\lceil \pi \sqrt{k}/2\rceil\le \pi \sqrt{k}/2+1$.

Next, consider the function $f\circ \PrOR_k$.
We can approximate this function to error $\epsilon$
simply by plugging in $n$ independent copies of the polynomial
$q$ into the variables of the polynomial $p$.
This means that the approximate degree of $f\circ\PrOR_k$
to error $\epsilon$ is at most
$\adeg_\epsilon(f)\cdot(\pi \sqrt{k}/2+1)$.

On the other hand, we now claim that for
appropriate choice of $k$, we have
$\bs(0^{kn},f\circ\PrOR_k)\ge k\fbs(0^n,f)$,
and hence $\bs(f\circ \PrOR_k)\ge k\fbs(f)$.
To see this, let $\{w_B\}_B$ be an optimal
weight scheme for the fractional block sensitivity
of $0^n$ with respect to $f$, so that
$\sum_{B:i\in B} w_B\le 1$ and $\sum_B w_B=\fbs(f)$.
Note that since fractional block sensitivity
is a linear program, the optimal solution can be taken
to be rational; let $L$ be a common denominator
of all the weights, so that $Lw_B$ is an integer for
each sensitive block $B$. Now take $k$ to be an
integer multiple of $L$. For each sensitive block $B$
of $0^n$ with respect to $f$, we define $kw_B$
different sensitive blocks of $0^{kn}$ with respect
to $f\circ\PrOR_k$, such that all of the new blocks
are mutually disjoint. To do so, we simply use
a different bit in each copy of $\PrOR_k$
for each block. Since the sum of weights $w_B$
for blocks that use bit $i$ of the input to $f$
is at most $1$, the total number of new blocks we will
generate that use copy $i$ of $\PrOR_k$ is at most
$k$, and hence we can give each block a different
bit of that copy of $\PrOR_k$. The total number
of disjoint blocks will then be $k\sum_B w_B=k\fbs(f)$.

We conclude that $\bs(f\circ \PrOR_k)\ge k\fbs(f)$
as long as $k$ is a multiple of a certain integer $L$.
Now, by a standard result \cite{BBC+01,BdW02},
we know that the approximate degree to error
$\epsilon$ of a (possibly partial)
Boolean function is at least the square root of its
block sensitivity; more explicitly, we have
\[\adeg_\epsilon(f\circ\PrOR_k)\ge
\sqrt{\frac{1-2\epsilon}{2(1-\epsilon)}\bs(f\circ \PrOR_k)}
\ge \sqrt{\frac{1-2\epsilon}{2(1-\epsilon)}k\fbs(f)}.\]
Combined with our upper bound on this degree, we have
\[\adeg_\epsilon(f)\cdot(\pi \sqrt{k}/2+1)
\ge \sqrt{\frac{1-2\epsilon}{2(1-\epsilon)}k\fbs(f)},\]
and since $k$ can go to infinity, we must have
\[\adeg_\epsilon(f)\ge
\frac{\sqrt{2}}{\pi}
\sqrt{\frac{(1-2\epsilon)}{1-\epsilon}\fbs(f)},\]
from which the desired result follows.
\end{proof}

\begin{theorem}\label{thm:adeg_cfbs}
For all (possibly partial) Boolean functions $f$, we have
\[\adeg_\epsilon(f)\ge\frac{\sqrt{(1-2\epsilon)\cfbs(f)}}{\pi}.\]
\end{theorem}

\begin{proof}
Let $p$ be a polynomial which approximates $f$ to error
$\epsilon$. Then $p(x)\in[0,1]$ for all $x\in\B^n$,
so define $f'(x)$ by $f'(x)=1$ if $p(x)\ge 1/2$ and
$f'(x)=0$ if $p(x)<1/2$. It is clear that $f'(x)=f(x)$
for all $x\in\Dom(f)$, so $f'$ is a completion of $f$.
Let $x\in\Dom(f)$ be an input so that
$\fbs(x,f')\ge\cfbs(f)$. To complete
the proof, it will suffice to lower bound
the degree of $p$ by $\Omega(\sqrt{\fbs(x,f')})$.

Suppose without loss of generality that $f(x)=0$
(otherwise, negate $f$ and $f'$ and replace $p$ with $1-p$).
Then we know that $p(x)\in[0,\epsilon]$,
and that for any $y\in\B^n$ such that $f'(x)\ne f'(y)$,
we have $p(y)\in[1/2,1]$. This means that the
polynomial $q(z)=(2p(z)+1-2\epsilon)/(3-2\epsilon)$
has the same degree as $p$, is bounded in $[0,1]$
on $\B^n$, and approximates $f'$ to error
$1/(3-2\epsilon)$ on the input $x$ and on all
inputs $y\in\B^n$ such that $f'(x)\ne f'(y)$.
In other words, consider the partial function $f'_x$
which is the restriction of $f'$ to the promise
set $\{x\}\cup\{y\in\B^n:f'(y)\ne f'(x)\}$.
Then $q$ approximates $f'_x$ to error $1/(3-2\epsilon)$,
and has the same degree as $p$. Now, it is not
hard to see that $\fbs(f'_x)=\fbs(x,f')$.
Hence it suffices to lower bound the degree of $q$
by $\Omega(\sqrt{\fbs(f'_x)})$. Such a lower bound follows
from \lem{adeg_fbs}; indeed, we conclude that
the degree of $p$ is at least
\[\frac{1}{\pi}\sqrt{\frac{1-2\epsilon}{1-\epsilon}\cfbs(f)}.
\qedhere\]
\end{proof}

\begin{theorem}
For all (possibly partial) Boolean functions $f$,
\[\CAdv(f)\le 2\Adv(f)^2.\]
\end{theorem}

\begin{proof}
Let $f$ be a (possibly partial) Boolean function,
and $q(x,i)$ be a feasible weight scheme for $\Adv(f)$ that has
$\sum_{i\in[n]} q(x,i)\le\Adv(f)$ for all $i$.
Fix any $x,y\in\Dom(f)$ such that $f(x)\ne f(y)$. Then
\[1\le \sum_{i:x_i\ne y_i}\sqrt{q(x,i)q(y,i)}
=\sum_{i:x_i\ne y_i}
\sqrt{\min\{q(x,i),q(y,i)\}\max\{q(x,i),q(y,i)\}}\]
\[\le \sqrt{\sum_{i:x_i\ne y_i}\min\{q(x,i),q(y,i)\}
\cdot\sum_{i:x_i\ne y_i}\max\{q(x,i),q(y,i)\}}.\]
Note that $\max\{q(x,i),q(y,i)\}\le q(x,i)+q(y,i)$,
and we know the sum over $i$ of $q(x,i)$ and $q(y,i)$
are each at most $\Adv(f)$. Hence we get
\[\sum_{i:x_i\ne y_i}\max\{q(x,i),q(y,i)\}\le 2\Adv(f),\]
and hence
\[\sum_{i:x_i\ne y_i}\min\{q(x,i),q(y,i)\}
\ge \frac{1}{2\Adv(f)}.\]
This means that if we scale the weights $q(x,i)$ up by
a uniform factor of $2\Adv(f)$, the resulting weight
scheme $q'(x,i)$ will be feasible for $\CAdv(f)$.
The objective value of this new weight scheme will then be
the maximum over $x$ of
\[\sum_{i\in[n]}q'(x,i)=2\Adv(f)\sum_{i\in[n]} q(x,i)
\le 2\Adv(f)^2,\]
so $\CAdv(f)\le 2\Adv(f)^2$, as desired.
\end{proof}

\section*{Acknowledgements}

We thank Rahul Jain and Dave Touchette
for helpful discussions related to
the $\QICZ(G)>0$ conjecture. We thank Robin Kothari
for helpful discussions related to the
adversary bounds. We thank Anne Broadbent
for helpful discussions related to quantum secure
2-party computation. We thank Mika G{\"o}{\"o}s
for helpful discussions regarding critical block
sensitivity and its lifting theorem. We thank
Jevg{\=e}nijs Vihrovs and
the other authors of \cite{AKPV18} for helpful discussions
regarding the classical adversary method, and
particularly Kri{\v s}j{\=a}nis Pr{\=u}sis
for the proof of \lem{CAdv_cfbs}.

A.A. is supported by the NSF QLCI Grant No. 2016245. S.B. is supported in part by the Natural Sciences and Engineering Research Council of Canada (NSERC), DGECR-2019-00027 and RGPIN-2019-04804\footnote{Cette recherche a été financée par le Conseil de recherches en sciences naturelles et en génie du Canada (CRSNG),
DGECR-2019-00027 et RGPIN-2019-04804.}. S.K. is supported by the National Research Foundation, including under NRF RF Award No. NRF-NRFF2013-13, the Prime Minister’s Office, Singapore; the Ministry of Education, Singapore, under the Research Centres of Excellence program and by Grant No. MOE2012-T3-1-009; and in part by the NRF2017-NRF-ANR004 VanQuTe Grant.
Part of this work was done when S.K. was visiting the Institute of Quantum Computing, University of Waterloo and when A.A. was affiliated to the Department of Combinatorics $\&$ Optimization and Institute for Quantum Computing, University of Waterloo, as well as to the Perimeter Institute for Theoretical Physics.

\phantomsection\addcontentsline{toc}{section}{References} 
\renewcommand{\UrlFont}{\ttfamily\small}
\let\oldpath\path
\renewcommand{\path}[1]{\small\oldpath{#1}}
\emergencystretch=1em 
\printbibliography

\end{document}